\documentclass[lettersize,journal]{IEEEtran}
\usepackage{amsmath,amsfonts}
\usepackage{algorithmic}
\usepackage{algorithm}
\usepackage{array}
\usepackage[caption=false,font=normalsize,labelfont=rm,textfont=rm]{subfig}
\usepackage{textcomp}	
\usepackage{stfloats}
\usepackage{url}
\usepackage{verbatim}
\usepackage{graphicx}
\usepackage{cite}
\usepackage{epstopdf}
\newtheorem{theorem}{Theorem}
\usepackage[
colorlinks=true,
linkcolor=blue, 
citecolor=blue, 
urlcolor=blue  
]{hyperref}

\usepackage[font=footnotesize,skip=4pt]{caption}

\begin{document}

\title{Capacity Analysis and Joint Gaussian Beam Pattern Optimization for Positioning-Assisted Communications}
	
\author{Yuanbo Liu,~\IEEEmembership{Student Member,~IEEE}, Rui Wang,~\IEEEmembership{Student Member,~IEEE}, Han Zhang, Shuojin Huang,~\IEEEmembership{Student Member,~IEEE}, Zaichen Zhang,~\IEEEmembership{Senior Member,~IEEE}, Bingcheng Zhu,~\IEEEmembership{Senior Member,~IEEE}
	
\thanks{Y. Liu, R. Wang, H. Zhang, S. Huang, Z. Zhang and B. Zhu are with the National Mobile Communications Research Laboratory, Frontiers Science Center for Mobile Information Communication and Security, Southeast University, Nanjing 210096, China. H. Zhang, Z. Zhang and B. Zhu are also with the Purple Mountain Laboratories, Nanjing 211111, China. }
\thanks{B. Zhu is the corresponding author (e-mail: zbc@seu.edu.cn).}}



\maketitle

\begin{abstract}
Positioning-assisted beamforming is a novel enabling technology for massive multiple-input multiple-output in modern wireless communications, as it provides reliable beam steering without requiring explicit channel state information. However, the channel capacity analysis under positioning error remains mathematically intractable, which limits both performance characterization and beamforming design. In this work, we derive instantaneous and ergodic channel capacity approximations in closed forms for both two-dimensional and three-dimensional positioning-assisted beamforming systems. Based on the new expressions, we provide closed-form optimal joint Gaussian beam pattern that maximizes the asymptotic ergodic capacity. Numerical results verify the theoretical capacity expressions and the optimal beam pattern. The derived expressions enable efficient beam design to maximize channel capacity and provide a theoretical basis for positioning-assisted beamforming design.
\end{abstract}


\section{Introduction}
\IEEEPARstart{M}{assive} multiple-input multiple-output technology can support high data rates and connection density~\cite{thallapalli2026gridless,zhang2025channel}. However, the severe path loss and poor penetration characteristics of high frequency signals introduce critical challenges to wide coverage. To mitigate these limitations, beamforming has been regarded as a fundamental technique~\cite{zhang2026fundamental,wang2021joint,zhang2024design}. It concentrates the transmitted energy onto the intended user by dynamically establishing, tracking, and adjusting directional beam. This process effectively compensates for propagation loss and reduces channel interference, which also improves overall system capacity and coverage performance.

However, the performance of conventional beamforming schemes strongly depends on accurate channel state information~(CSI), which incurs substantial training overhead and signaling latency. In frequency division duplex systems, CSI estimation requires extensive downlink pilot transmission and uplink feedback, resulting in considerable communication overhead and reduced spectral efficiency~\cite{tomasin2025two}. In time division duplex systems, CSI estimation depends on channel reciprocity calibration, and the calibration accuracy can degrade rapidly in environments with fast channel variations. Frequent CSI updates and reciprocity calibration are required to maintain beamforming performance, which further increases system overhead and latency~\cite{becirovic2022combining}. Consequently, CSI estimation becomes a critical challenge for beamforming implementations.

To overcome the challenges associated with complex CSI acquisition, positioning-assisted beamforming has emerged as a promising alternative strategy~\cite{talvitie2019positioning}. The central idea of this approach is to exploit large-scale geometric information, such as the spatial position and angle of the user, to directly determine the optimal beam directions, so as to relieve the  channel estimation~\cite{zhu2022outage,liu2026joint}. However, existing positioning-assisted beamforming systems are still subject to several practical limitations. Firstly, some methods use positioning information to reduce pilots overhead, but still rely on beam training, which results in pilot contamination and reduces the number of simultaneously supported users~\cite{zhang2026positioning,muppirisetty2018location}. To totally circumvent the requirement for CSI estimation, a CSI-free beamforming method was proposed based on receiver position~\cite{maiberger2010location}, which avoids explicit CSI estimation and feedback, but lacks theoretical analysis and relies on numerical simulations. Secondly, existing CSI-free beamforming methods still lack sufficient theoretical guidance for beam pattern design, making practical beam configuration largely dependent on empirical design, which may result in beam misalignment, service interruption, or unnecessary energy consumption~\cite{he2026bistatic}. To address this issue, closed-form capacity expressions were derived in~\cite{cao2026location}, while the results relied on oversimplified positioning error models that may not accurately capture practical positioning errors. The impact of positioning errors on channel capacity was investigated in~\cite{talvitie2020beamformed}, demonstrating that beam optimization is essential for maximizing channel capacity under a given positioning accuracy. However, closed-form optimal beam design was not established, which limits the direct application of these analytical results to practical beam pattern design. Thirdly, although closed-form optimal beam patterns were subsequently derived in~\cite{song2024position,jing2026beamforming} to maximize spectral efficiency, these results are limited to two-dimensional~(2D) scenarios. Three-dimensional~(3D) positioning-assisted beam design and spectral efficiency optimization were investigated in~\cite{shi2023joint,srinivasan2021airplane,lu2020positioning}, but the corresponding beam design relied on numerical optimization, making it computationally demanding. So far, it still remains an open fundamental problem to maximize the channel capacity through simple beam pattern design in 2D and 3D scenarios. 

In this work, we investigate the channel capacity of positioning-assisted beamforming systems and develop the capacity-maximizing beam pattern. The contributions of this work can be summarized as follows:
\begin{itemize}
	\item We develop a unified channel model for positioning-assisted beamforming by jointly characterizing positioning errors and beam patterns in both 2D and 3D scenarios.
	\item Based on the proposed channel model, we derive the closed-form expressions of instantaneous and ergodic capacity, considering the link distance, transmit power, positioning error and beam pattern.
	\item Based on the derived ergodic capacity expressions, the closed-form optimal beam pattern is developed to maximize the ergodic capacity in both 2D and 3D scenarios.
\end{itemize}

This paper is organized as follows. In Section~\ref{System Model}, we present a system model for positioning-assisted beamforming systems considering the positioning errors. In Section~\ref{Capacity Derivation and Analysis}, we derive closed-form ergodic channel capacity in both 2D and 3D scenarios. In Section~\ref{Beam Pattern Optimization}, optimal beam pattern for maximizing the ergodic capacity is derived in closed form. In Section~\ref{Numerical Results}, numerical results of the analytical results are presented. Finally, Section~\ref{Conclusion} draws some concluding remarks.

\section{System Model}\label{System Model}
In a typical positioning-assisted beam transmitting model, the received power $P_r$ can be expressed as~\cite[eq.~(1)]{friis1946note}
\begin{equation}\label{Eq_Pr}
	P_r = P_{\max} A_e G_\theta G_d
\end{equation}
where $P_{\max}$ is the peak power density of the transmitter antenna; $A_e$ is the effective area of the receiver antenna defined as~\cite[eq.~(2-110)]{milligan2005modern}
\begin{equation}
	A_e = \frac{G_r \lambda^2}{4 \pi}
\end{equation}
where $G_r$ is the receiver antenna gain and $\lambda$ is the wavelength; $G_\theta \leq 1$ represents the normalized transmitter antenna gain; $G_d$ is the free space propagation loss, which can be calculated as~\cite[eq.~(2-11)]{milligan2005modern}
\begin{equation}\label{Eq_Gd}
	G_d = \frac{1}{4 \pi d^2 }
\end{equation}
where $d$ is the link distance between the transmitter and receiver antenna.

\subsection{2D Positioning-Assisted Beamforming Model}
For the 2D beam model, $G_\theta$ in~\eqref{Eq_Pr} can be represented as~\cite[eq.~(3.42)]{skolnik1980introduction}
\begin{equation}\label{Eq_2D_Gtheta}
	G_\theta = 10^{-1.2 \frac{\theta^2}{\theta_{3 \text{dB}}^2}}
\end{equation}
where $\theta$ denotes the error angle between the antenna boresight and the direction of arrival, and $\theta_{3\text{dB}}$ represents the 3-dB beamwidth of the antenna pattern. In practical implementations, the transmit power $P_t$ remains constant and can be related to $P_{\max}$ as
\begin{equation}\label{Eq_2D_Pt}
	P_t = \int_{-\pi}^{\pi} P_{\max} G_\theta d\theta \approx P_{\max} \theta_{3\text{dB}}\sqrt{ \frac{\pi}{1.2 \ln 10}}
\end{equation}
and $P_{\max}$ can be approximated as
\begin{equation}\label{Eq_2D_Pmax}
	P_{\max} \approx \frac{P_t}{\theta_{3\text{dB}}} \sqrt{\frac{1.2 \ln 10}{\pi}}.
\end{equation}
A typical positioning-assisted beam model is shown in Fig.~\ref{Figure_2D_Beam_System_Model}. Without loss of generality, the receiver is located at $\mathbf{p}_u^{2\text{D}} = \left[ 0, d \right]^T$ and its estimated position is $\hat{\mathbf{p}}_u^{2\text{D}} = \left[ \hat{x}_u, \hat{y}_u \right]^T$. The transmitter is located at the origin and the error angle $\theta$ can be calculated as
\begin{equation}\label{Eq_2D_theta}
	\theta = \mathrm{atan2} \left( \hat{x}_u, \hat{y}_u  \right) 
\end{equation}
where $\mathrm{atan2} \left( \cdot, \cdot \right)$ is the four-quadrant inverse tangent function. 
\begin{figure}
	\centering
	\includegraphics[trim=2pt 1pt 2pt 1pt,clip,width=0.4\textwidth]{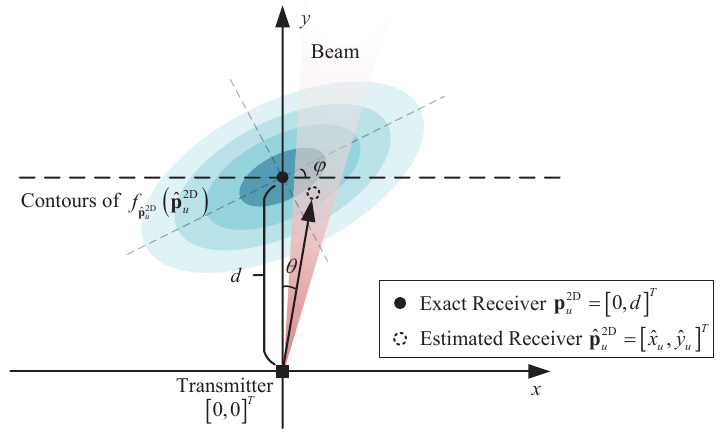}\\
	\caption{Typical 2D positioning-assisted beam transmitting model. The transmitter is located at the origin, and the exact position is $\mathbf{p}_u^{2\text{D}} = \left[0, d\right]^T$; the estimated position is $\hat{\mathbf{p}}_u^{2\text{D}} = \left[\hat{x}_u, \hat{y}_u\right]^T$. The ellipses represent the contours of PDF $f_{\hat{\mathbf{p}}_u^{2\text{D}}} \left( \hat{\mathbf{p}}_u^{2\text{D}} \right)$ and the beam is directed toward the estimated receiver position.
	}\label{Figure_2D_Beam_System_Model}
\end{figure}
It is shown that $\hat{\mathbf{p}}_u^{2\text{D}}$ follows a bivariate Gaussian distribution with its probability density function~(PDF) denoted by~\cite{hyun2021adaptive}
\begin{equation}\label{Eq_2D_PDF}
	\begin{aligned}
		&f_{\hat{\mathbf{p}}_u^{2\text{D}}} \left( \hat{\mathbf{p}}_u^{2\text{D}} \right) = \frac{1}{2 \pi \left|\mathbf{\Sigma}^{2\text{D}}\right|^{\frac{1}{2}}} \\
		&\times \exp \left( -\frac{ \left( \hat{\mathbf{p}}_u^{2\text{D}} - \mathbf{p}_u^{2\text{D}} \right)^T \left( \mathbf{\Sigma}^{2\text{D}} \right)^{-1} \left( \hat{\mathbf{p}}_u^{2\text{D}} - \mathbf{p}_u^{2\text{D}} \right) }{2} \right)
	\end{aligned}
\end{equation}
where $\mathbf{\Sigma}^{2\text{D}} = E \left[ \left( \hat {\mathbf{p}}_u^{2\text{D}} - \mathbf{p}_u^{2\text{D}} \right) \left( \hat {\mathbf{p}}_u^{2\text{D}} - \mathbf{p}_u^{2\text{D}} \right)^T \right] \in \mathbb{R}^{2 \times 2}$ denotes the 2D covariance matrix of the positioning error and can be diagonalized through the eigenvalue decomposition as
\begin{equation}
	\begin{aligned}
		\mathbf{\Sigma}^{2\text{D}} & = \mathbf{R}^{2\text{D}} \mathbf{\Lambda}^{2\text{D}} \left( \mathbf{R}^{2\text{D}}\right)^{T} \\
		& =
		\begin{bmatrix}
			\cos\varphi & -\sin\varphi \\
			\sin\varphi & \cos\varphi
		\end{bmatrix}
		\begin{bmatrix}
			\sigma_1^2 & 0 \\
			0 & \sigma_2^2
		\end{bmatrix}
		\begin{bmatrix}
			\cos\varphi & \sin\varphi \\
			-\sin\varphi & \cos\varphi
		\end{bmatrix}
	\end{aligned}
\end{equation}
where $\mathbf{R}^{2\text{D}}$ denotes the 2D rotation matrix associated with the directional angle $\varphi$; $\sigma_1^2$ and $\sigma_2^2$ represent the error variances along the major and minor axes of the PDF contours.

\subsection{3D Positioning-Assisted Beamforming Model}
In 3D cases, the normalized antenna gain is jointly determined by the horizontal and vertical components. Accordingly, eq.~\eqref{Eq_2D_Gtheta} can be generalized to~\cite{skolnik1980introduction}
\begin{equation}\label{Eq_3D_Gtheta_phi}
	G_{\left( \theta, \phi \right)} =  10^{-1.2 \left[ \theta, \phi \right] \mathbf{A} \left[ \theta, \phi \right]^T}
\end{equation}
where $\theta$ and $\phi$ denote the angular deviation between the antenna boresight and the direction of arrival in the horizontal and vertical directions; $\mathbf{A}$ is a positive definite matrix defined as
\begin{equation}\label{Eq_3D_A}
	\mathbf{A} = 
	\begin{bmatrix}
		\frac{1}{\theta_{3 \text{dB}}^2} & m \\
		m & \frac{1}{\phi_{3 \text{dB}}^2}
	\end{bmatrix},
\end{equation}
which represents the 3D beam pattern in terms of its principal axes and orientation, where $\theta_{3 \text{dB}}$ and $\phi_{3 \text{dB}}$ denote the horizontal and vertical 3-dB beamwidth; $m$ represents the beam rotation. In 3D cases, the transmit power $P_t$ is computed as
\begin{equation}\label{Eq_3D_Pt}
	\begin{aligned}
		P_t &= \int_{-\pi / 2}^{\pi / 2} \int_{-\pi / 2}^{\pi / 2} P_{\max} G_{\left( \theta, \phi \right)} \, d\theta \, d\phi \\
		& \approx \frac{\pi \theta_{3\text{dB}} \phi_{3\text{dB}} P_{\max} }{1.2 \ln 10 \sqrt{1 - m^2 \theta_{3\text{dB}}^2 \phi_{3\text{dB}}^2 }},
	\end{aligned}
\end{equation}
and $P_{\max}$ can be further approximated as
\begin{equation}\label{Eq_3D_Pmax}
	P_{\max} \approx \frac{1.2 \ln 10 P_t \sqrt{1 - m^2 \theta_{3\text{dB}}^2 \phi_{3\text{dB}}^2 }}{\pi \theta_{3\text{dB}} \phi_{3\text{dB}}}.
\end{equation}
Figure~\ref{Figure_3D_Beam_System_Model} exhibits a general 3D positioning-assisted beam model. Assuming that the transmitter is located at the origin, the exact receiver and its estimated position are $\mathbf{p}_u^{3\text{D}} = \left[0, d, 0 \right]^T$ and $\hat{\mathbf{p}}_u^{3\text{D}} = \left[\hat{x}_u, \hat{y}_u, \hat{z}_u \right]^T$, where $d$ is the link distance. $\theta$ and $\phi$ can be calculated as
\begin{equation}\label{Eq_3D_theta}
	\theta  = \mathrm{atan2} \left( \hat{x}_u, \hat{y}_u \right) \in \left( -\pi, \pi \right]
\end{equation}
and
\begin{equation}\label{Eq_3D_phi}
	\phi = \arctan \frac{\hat{z}_u}{\sqrt{\hat{x}_u^2 + \hat{y}_u^2}} \in \left( -\pi/2, \pi/2 \right].
\end{equation}
Since $\hat{\mathbf{p}}_u^{3\text{D}}$ follows a trivariate Gaussian distribution, its PDF is given by
\begin{equation}\label{Eq_3D_PDF}
	\begin{aligned}
		&f_{\hat{\mathbf{p}}_u^{3\text{D}}} \left( \hat{\mathbf{p}}_u^{3\text{D}} \right) = \frac{1}{\left( 2 \pi \right)^{\frac{3}{2}} \left|\mathbf{\Sigma}^{3\text{D}}\right|^{\frac{1}{2}}} \\
		&\times \exp \left( -\frac{ \left( \hat{\mathbf{p}}_u^{3\text{D}} - \mathbf{p}_u^{3\text{D}} \right)^T \left( \mathbf{\Sigma}^{3\text{D}} \right)^{-1} \left( \hat{\mathbf{p}}_u^{3\text{D}} - \mathbf{p}_u^{3\text{D}} \right) }{2} \right)
	\end{aligned}
\end{equation}
where $\mathbf{\Sigma}^{3\text{D}} = E \left[ \left( \hat {\mathbf{p}}_u^{3\text{D}} - \mathbf{p}_u^{3\text{D}} \right) \left( \hat {\mathbf{p}}_u^{3\text{D}} - \mathbf{p}_u^{3\text{D}} \right)^T \right] \in \mathbb{R}^{3 \times 3}$ represents the 3D covariance matrix of the positioning error and can be expressed by spectral decomposition as $\mathbf{\Sigma}^{3\text{D}} = \mathbf{R}^{3\text{D}} \mathbf{\Lambda}^{3\text{D}} \left( \mathbf{R}^{3\text{D}} \right)^T$ where  $\mathbf{R}^{3\text{D}} = \mathbf{R}_z \left( \varphi_z \right) \mathbf{R}_y \left( \varphi_y \right) \mathbf{R}_x \left( \varphi_x \right)$ represents the 3D rotation matrix parameterized by directional angle $\varphi_x, \varphi_y, \varphi_z$ with $\mathbf{R}_x \left( \cdot \right), \mathbf{R}_y \left( \cdot \right), \mathbf{R}_z \left( \cdot \right)$ denoting the elementary rotation matrices about the $x$, $y$, $z$ axes, respectively~\cite{mahony2012multirotor}. $\mathbf{\Lambda}^{3\text{D}} = \mathrm{diag}(\sigma_1^2, \sigma_2^2, \sigma_3^2)$ contains the eigenvalues with $\sigma_1^2 > \sigma_2^2 > \sigma_3^2$, which determine the variances of the positioning errors along different axes of the PDF ellipsoidal contour surfaces.

\begin{figure}
	\centering
	\includegraphics[trim=2pt 1pt 2pt 1pt,clip,width=0.4\textwidth]{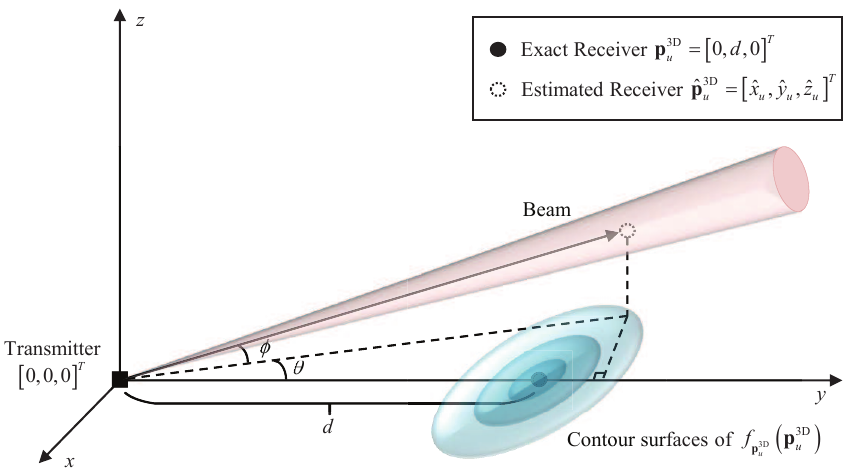}\\
	\caption{General 3D positioning-assisted beam transmitting model. The transmitter is located at the origin, the exact position is $\mathbf{p}_u^{3\text{D}} = \left[0, d, 0 \right]^T$; the estimated position is $\hat{\mathbf{p}}_u^{3\text{D}} = \left[\hat{x}_u, \hat{y}_u, \hat{z}_u \right]^T$. The ellipses represent the contour surfaces of PDF $f_{\hat{\mathbf{p}}_u^{3\text{D}}} \left( \hat{\mathbf{p}}_u^{3\text{D}} \right)$ and the beam is directed toward the estimated receiver position.
	}\label{Figure_3D_Beam_System_Model}
\end{figure}

\section{2D Asymptotic Capacity Expressions}\label{Capacity Derivation and Analysis}
In this section, we derive the channel capacity of the positioning-assisted beamforming system in 2D cases.

\subsection{2D Ergodic Capacity in an Integral Form}
The instantaneous capacity for complex baseband channel of 2D positioning-assisted systems can be expressed by substituting~\eqref{Eq_Pr}-\eqref{Eq_2D_Gtheta} into the Shannon capacity equation as
\begin{equation}\label{Eq_2D_Cinst_theta}
	C_{\text{inst}}^{2\text{D}} = \log_2 \left( 1 + \frac{P_r}{N_0} \right) = \log_2 \left( 1 + \frac{P_{\max} A_e}{4 \pi d^2 N_0}\cdot 10^{-1.2 \frac{\theta^2}{\theta_{3 \text{dB}}^2}} \right)
\end{equation}
where $N_0$ denotes the power of the additive Gaussian noise. Since the positioning error is much smaller than the link distance, we can assume $(\hat{y}_u - d) / d \rightarrow 0$. Applying dual-variable Taylor expansion to $\theta$ in~\eqref{Eq_2D_theta} at $\left( \hat{x}_u, \hat{y}_u \right) = \left( 0,d \right)$, we can simplify $\theta$ in~\eqref{Eq_2D_Cinst_theta} as
\begin{equation}\label{Eq_2D_theta_Taylor}
	\theta = \mathrm{atan2} \left( \hat{x}_u, \hat{y}_u  \right) = \frac{\hat{x}_u}{d} + o \left( || \left( \hat{x}_u,  \left( \hat{y}_u - d \right) \right) || \right)
\end{equation}
when $\left( \hat{x}_u, \hat{y}_u \right) \to \left( 0,d \right)$. Substituting~\eqref{Eq_2D_theta_Taylor} into~\eqref{Eq_2D_Cinst_theta}, we can approximate $C_{\text{inst}}^{2\text{D}}$ as
\begin{equation}\label{Eq_2D_Cinst_x}
	C_{\text{inst}}^{2\text{D}} \! \left( \hat{x}_u \right) \! = \! \log_2 \! \left( \! 1 \! + \! \frac{P_{\max} A_e}{ 4 \pi d^2 N_0} \! \cdot \! 10^{\! -1.2 \frac{\hat{x}_u^2}{d^2 \! \theta_{3 \text{dB}}^2} +  o \left( \! ||   \left( \hat{x}_u,  \left( \hat{y}_u - d \right) \! \right)  ||^{\! 2} \! \right) } \!\! \right).
\end{equation}
Combining~\eqref{Eq_2D_PDF} and~\eqref{Eq_2D_Cinst_x}, the ergodic channel capacity can be calculated by marginalizing $\hat{y}_u$ as 
\begin{equation}\label{Eq_2D_Cerg_exact}
	\begin{aligned}
		C_{\text{erg}}^{2\text{D}} \!\! = \!\!\! \iint_{-\infty}^{+\infty} \!\!\! C_{\text{inst}}^{2\text{D}} \! \left( x \right) \! \cdot \!\! f_{\hat{\mathbf{p}}_u^{2\text{D}}} \! \left( \hat{\mathbf{p}}_u^{2\text{D}} \right) \! dxdy \! = \!\! \int_{-\infty}^{+\infty} \!\!\! C_{\text{inst}}^{2\text{D}} \! \left( x \right) \! \cdot \!\! f_{\hat{x}_u} \!\! \left( x \right) \! dx
	\end{aligned}
\end{equation}
where $f_{\hat{x}_u} \left( x \right)$ is the marginal PDF of $\hat{x}_u$ as
\begin{equation}\label{Eq_2D_fx}
	f_{\hat{x}_u} \left( x \right) = \frac{1}{\sqrt{2 \pi \sigma_x^2}} \exp \left( -\frac{x^2}{2\sigma_x^2} \right)
\end{equation}
and $\sigma_x^2$ is the variance of  $\hat{x}_u$.

\subsection{Asymptotic Conditions and Capacity Expressions}\label{Two Asymptotic Cases}
Since the ergodic capacity expression in~\eqref{Eq_2D_Cerg_exact} is difficult to evaluate analytically, we present two expressions of~\eqref{Eq_2D_Cerg_exact} under two asymptotic cases and provide their proofs.

\subsubsection{Case 1: $\theta_{3 \rm{dB}} d / \sigma_x \to 0$}
This asymptotic case corresponds to a narrow beamwidth $\theta_{3 \text{dB}}$ much smaller than $\sigma_x / d$.
\begin{theorem}\label{Theorem_Asymptotic_Case1}
	Under the asymptotic case $\theta_{3 \text{dB}} d / \sigma_x \to 0$, the ergodic capacity converges to
	\begin{equation}\label{Eq_2D_Asymptotic_Case1_narrow_beam}
		\lim_{\frac{\theta_{3 \text{dB}} d}{\sigma_x} \to 0} C_{\text{erg}}^{2\text{D}} =  f_{\hat{x}_u} (0) \cdot  \int_{-\infty}^{+\infty} C_{\text{inst}}^{2\text{D}} \left( x \right) dx.
	\end{equation}
\end{theorem}
The proof of Theorem~\ref{Theorem_Asymptotic_Case1} is given in Appendix~\ref{Appendix_Theorem1}. Theorem~\ref{Theorem_Asymptotic_Case1} shows that, for sufficiently narrow beamwidth, $C_{\text{inst}}^{2\text{D}} \left( x \right)$ is almost within a small neighborhood of the origin, and behaves like the sampling pulse, and thus the PDF function $f_{\hat{x}_u} \left( x \right)$ can be approximated by $f_{\hat{x}_u} \left( 0 \right)$.

\subsubsection{Case 2: $\theta_{3 \rm{dB}} d / \sigma_x \to +\infty$}
This asymptotic case corresponds to negligible positioning error and the beamwidth $\theta_{3 \text{dB}}$ is much larger than $\sigma_x / d$.
\begin{theorem}\label{Theorem_Asymptotic_Case2}
	Under the asymptotic case $\theta_{3 \text{dB}} d / \sigma_x \to +\infty$, the ergodic capacity converges to
	\begin{equation}\label{Eq_2D_Asymptotic_Case2_wide_beam}
		\lim_{\frac{\theta_{3 \text{dB}} d}{\sigma_x} \to +\infty} C_{\text{erg}}^{2\text{D}} =  C_{\text{inst}}^{2\text{D}} \left( 0 \right) \cdot  \int_{-\infty}^{+\infty}  f_{\hat{x}_u} (x) dx = C_{\text{inst}}^{2\text{D}} \left( 0 \right).
	\end{equation}
\end{theorem}
The proof of Theorem~\ref{Theorem_Asymptotic_Case2} is given in Appendix~\ref{Appendix_Theorem2}.
Theorem~\ref{Theorem_Asymptotic_Case2} indicates that, when the beamwidth is sufficiently wide, the beam nearly covers the possible range of the positioning error, and $C_{\text{inst}}^{2\text{D}} \left( x \right)$ is nearly flat around the origin, which can be approximated by $C_{\text{inst}}^{2\text{D}} \left( 0 \right)$.

\subsection{Approximations to $C_{\rm{inst}}^{2\rm{D}}$}
In Section~\ref{Two Asymptotic Cases}, two simplified expressions of~\eqref{Eq_2D_Cerg_exact} were derived for the two asymptotic cases, but closed-form expression of $C_{\text{erg}}^{2\text{D}}$ in~\eqref{Eq_2D_Cerg_exact} remains difficult to obtain. In this section, we propose an approximation of $C_{\text{inst}}^{2\text{D}} \left( x \right)$ in~\eqref{Eq_2D_Cerg_exact}, which also satisfies the two asymptotic cases in~\eqref{Eq_2D_Asymptotic_Case1_narrow_beam} and~\eqref{Eq_2D_Asymptotic_Case2_wide_beam}, and results in closed-form $C_{\text{erg}}^{2\text{D}}$ expression.

First, we perform the second-order Taylor expansion of $C_{\text{inst}}^{2\text{D}} \left( \hat{x}_u \right)$ in~\eqref{Eq_2D_Cinst_x} at $\hat{x}_u = 0$ as
\begin{equation}\label{Eq_2D_Cinst_2nd}
	\begin{aligned}
		C_{\text{inst}}^{2\text{D},2\text{nd}} \left( \hat{x}_u \right) &=  \log_2  \left(  1  + \frac{P_{\max} A_e}{ 4 \pi d^2 N_0}  \right) \\
		&- \frac{1.2 \log_2 10 \, \hat{x}_u^2}{d^2 \theta_{3 \text{dB}}^2} \cdot \frac{\frac{P_{\max} A_e}{ 4 \pi d^2 N_0}}{\frac{P_{\max} A_e}{ 4 \pi d^2 N_0} + 1} + o \left(  \frac{\hat{x}_u^2}{d^2} \right)
	\end{aligned}
\end{equation}
when $\hat{x}_u \to 0$. Noting that~\eqref{Eq_2D_Cinst_2nd} tends to $-\infty$ as $\hat{x}_u \to \infty$, while the exact capacity in~\eqref{Eq_2D_Cinst_x} remains non-negative, we clamp the negative part of~\eqref{Eq_2D_Cinst_2nd}  and obtain a non-negative approximation to~\eqref{Eq_2D_Cinst_x} as
\begin{equation}\label{Eq_2D_Cinst_2nd_clamp}
	\begin{aligned}
		&C_{\text{inst}}^{2\text{D},2\text{nd},+} \left( \hat{x}_u \right) =  \\
		&\begin{cases}
			\log_2  \left( 1 +  \frac{P_{\max} A_e}{ 4 \pi d^2 N_0} \! \right)  -  \frac{1.2  \log_2  10 \, \hat{x}_u^2}{d^2 \theta_{3 \text{dB}}^2} \! \cdot \! \frac{\frac{P_{\max} A_e}{ 4 \pi d^2 N_0}}{\frac{P_{\max} A_e}{ 4 \pi d^2 N_0} + 1}, \left| \hat{x}_u \right|  \leq   x_0 \\
			 0, \qquad \qquad \qquad \qquad \qquad \qquad \qquad \qquad \quad  \;\;\; \left| \hat{x}_u \right|  >  x_0
		\end{cases}
	\end{aligned}
\end{equation}
where $x_0 > 0$ is the positive root of~\eqref{Eq_2D_Cinst_2nd} satisfying
\begin{equation}\label{Eq_2D_x0_definition}
	\log_2 \left( 1 + \frac{P_{\max} A_e}{ 4 \pi d^2 N_0} \right)  - \frac{1.2 \log_2  10 \, x_0^2}{d^2 \theta_{3 \text{dB}}^2} \cdot \frac{\frac{P_{\max} A_e}{ 4 \pi d^2 N_0}}{\frac{P_{\max} A_e}{ 4 \pi d^2 N_0} + 1} = 0,
\end{equation}
and can be calculated as
\begin{equation}\label{Eq_2D_x0}
	x_0 = d \theta_{3 \text{dB}} \left( \frac{\log_2 \left( 1 + \frac{P_{\max} A_e}{ 4 \pi d^2 N_0} \right) }{1.2 \log_2 10} \cdot \frac{\frac{P_{\max} A_e}{ 4 \pi d^2 N_0} + 1}{\frac{P_{\max} A_e}{ 4 \pi d^2 N_0}} \right)^{\frac12} .
\end{equation}
However, eq.~\eqref{Eq_2D_Cinst_2nd_clamp} are uniquely determined by the original function $C_{\text{inst}}^{2\text{D}} \left( \hat{x}_u \right)$ in~\eqref{Eq_2D_Cinst_x} and the Taylor expansion point $\hat{x}_u = 0$. Hence, no additional degrees of freedom are available for~\eqref{Eq_2D_Cinst_2nd_clamp} to satisfy the two asymptotic cases in~\eqref{Eq_2D_Asymptotic_Case1_narrow_beam},~\eqref{Eq_2D_Asymptotic_Case2_wide_beam}. To solve this problem, based on the  non-negative second-order Taylor expansion in~\eqref{Eq_2D_Cinst_2nd_clamp}, we further introduce a fourth-order term and propose a polynomial expansion technique for the integrand factor $C_{\text{inst}}^{2\text{D}} \left( \hat{x}_u \right)$ in~\eqref{Eq_2D_Cinst_x} as
\begin{equation}\label{Eq_2D_Cinst_infty}
	\begin{aligned}
		&C_{\text{inst}}^{2\text{D},+} \left( \hat{x}_u \right) =  \\
		&\!\begin{cases}
			\!\log_2  \left( \! 1  \!+ \!  \frac{P_{\max} A_e}{ 4 \pi d^2 \! N_0} \! \right) \! - \! \frac{1.2 \log_2 \! 10 \, \hat{x}_u^2}{d^2 \theta_{3 \text{dB}}^2} \! \cdot \! \frac{\frac{P_{\max} A_e}{ 4 \pi d^2 N_0}}{\frac{P_{\max} A_e}{ 4 \pi d^2 N_0} +  1} + \kappa \hat{x}_u^4 , \left| \hat{x}_u \right| \! \leq \!  x_0 \\
			\! 0, \qquad \qquad \qquad \qquad \qquad \qquad \qquad \qquad \qquad \quad \;\;\, \left| \hat{x}_u \right| \! > \! x_0
		\end{cases}
	\end{aligned}
\end{equation}
where $\kappa$ is a coefficient to provide an additional degree of freedom to make $C_{\text{inst}}^{2\text{D},+} \left( \hat{x}_u \right)$ satisfy the two asymptotic cases as
\begin{equation}\label{Eq_2D_Asymptotic_Case1_narrow_beam_C+}
	\lim_{\frac{\theta_{3 \text{dB}} d}{\sigma_x} \to 0} C_{\text{erg}}^{2\text{D}} =  f_{\hat{x}_u} (0) \cdot  \int_{-\infty}^{+\infty} C_{\text{inst},+}^{2\text{D}} \left( x \right) dx
\end{equation}
and
\begin{equation}\label{Eq_2D_Asymptotic_Case2_wide_beam_C+}
	\lim_{\frac{\theta_{3 \text{dB}} d}{\sigma_x} \to +\infty} C_{\text{erg}}^{2\text{D}}  = C_{\text{inst}}^{2\text{D},+} \left( 0 \right).
\end{equation}

\subsection{Calculation of $\kappa$ and Asymptotic Approximation to $C_{\rm{erg}}^{2\rm{D}}$}
In this section, we use the two asymptotic cases in Section~\ref{Two Asymptotic Cases} to calculate $\kappa$ and further obtain closed-form asymptotic approximation to $C_{\rm{erg}}^{2\rm{D}}$.

First, we need to prove that $C_{\text{inst}}^{2\text{D},+} \left( \hat{x}_u \right)$ in~\eqref{Eq_2D_Cinst_infty} can also satisfy the two asymptotic cases in Section~\ref{Two Asymptotic Cases}. Since $C_{\text{inst}}^{2\text{D},+} \left( 0 \right) = C_{\text{inst}}^{2\text{D}} \left( 0 \right)$ always holds, substitute $C_{\text{inst}}^{2\text{D},+} \left( 0 \right)$ into~\eqref{Eq_2D_Asymptotic_Case2_wide_beam} to replace $C_{\text{inst}}^{2\text{D}} \left( 0 \right)$, and $C_{\text{inst}}^{2\text{D},+} \left( \hat{x}_u \right)$ in~\eqref{Eq_2D_Cinst_infty} naturally satisfies the asymptotic condition of Case 2 in~\eqref{Eq_2D_Asymptotic_Case2_wide_beam}.

To ensure that $C_{\text{inst}}^{2\text{D},+} \left( \hat{x}_u \right)$ in~\eqref{Eq_2D_Cinst_infty} also satisfies Case 1 in~\eqref{Eq_2D_Asymptotic_Case1_narrow_beam}, we impose the following constraint as
\begin{equation}\label{Eq_2D_Cinst=Cinst_infty}
	\int_{-\infty}^{+\infty} C_{\text{inst}}^{2\text{D}} \left( x \right) dx = \int_{-x_0}^{x_0} C_{\text{inst}}^{2\text{D},+} \left( x \right) dx.
\end{equation}
Substituting~\eqref{Eq_2D_Cinst_x} and~\eqref{Eq_2D_Cinst_infty} into~\eqref{Eq_2D_Cinst=Cinst_infty} to replace $C_{\text{inst}}^{2\text{D}} \left( x \right)$ and $C_{\text{inst}}^{2\text{D},+}  \left( x \right)$, and rewriting~\eqref{Eq_2D_Cinst=Cinst_infty} as
\begin{equation}\label{Eq_2D_x0_to_solve}
	\begin{aligned}
		&\int_{-\infty}^{+\infty} \log_2 \left( 1 + \frac{P_{\max} A_e}{ 4 \pi d^2 N_0}\cdot 10^{-1.2 \frac{x^2}{d^2\theta_{3 \text{dB}}^2}} \right) dx \\
		=&\int_{-x_0}^{x_0} \left[  \log_2  \left( 1 +  \frac{P_{\max} A_e}{ 4 \pi d^2 N_0} \right) \right. \\
		&\qquad \,\, \left. - \frac{1.2 \log_2 10 \, x^2}{d^2 \theta_{3 \text{dB}}^2} \cdot \frac{\frac{P_{\max} A_e}{ 4 \pi d^2 N_0}}{\frac{P_{\max} A_e}{ 4 \pi d^2 N_0}  + 1} + \kappa  x^4 \right] dx.
	\end{aligned}
\end{equation}
The two sides of~\eqref{Eq_2D_x0_to_solve} can be evaluated separately. The right-hand side is first expressed as
\begin{equation}\label{Eq_2D_x0_right_hand_side}
	\begin{aligned}
		&\int_{-x_0}^{x_0} \left[  \log_2  \left( 1 +  \frac{P_{\max} A_e}{ 4 \pi d^2 N_0} \right) \right. \\
		&\qquad \,\, \left. - \frac{1.2 \log_2 10 \, x^2}{d^2 \theta_{3 \text{dB}}^2} \cdot \frac{\frac{P_{\max} A_e}{ 4 \pi d^2 N_0}}{\frac{P_{\max} A_e}{ 4 \pi d^2 N_0}  + 1} + \kappa  x^4 \right] dx \\
		&= 2 \log_2 \left( 1 + \frac{P_{\max} A_e}{ 4 \pi d^2 N_0} \right) x_0 \\
		&-\frac{2}{3} \cdot \frac{1.2 \log_2 10}{d^2 \theta_{3 \text{dB}}^2} \cdot \frac{\frac{P_{\max} A_e}{ 4 \pi d^2 N_0}}{\frac{P_{\max} A_e}{ 4 \pi d^2 N_0}  + 1} x_0^3 +\frac{2}{5} \kappa x_0^5.
	\end{aligned}
\end{equation}
The left-hand side of~\eqref{Eq_2D_x0_to_solve} can be expressed as
\begin{equation}\label{Eq_2D_x0_left_hand_side}
	\begin{aligned}
		&\int_{-\infty}^{+\infty} \log_2 \left( 1 + \frac{P_{\max} A_e}{ 4 \pi d^2 N_0}\cdot 10^{-1.2 \frac{x^2}{d^2\theta_{3 \text{dB}}^2}} \right) dx \\
		&= \frac{1}{\ln 2} \cdot \frac{d\theta_{3 \text{dB}}}{\sqrt{1.2 \ln 10}} \int_{-\infty}^{+\infty} \ln \left( 1 + \frac{P_{\max} A_e}{ 4 \pi d^2 N_0} \cdot  e^{-x^2} \right) dx.
	\end{aligned}
\end{equation}
For a given constant $K > 0$, we have
\begin{equation}\label{Eq_2D_Li}
	\begin{aligned}
		&\!\int_{-\infty}^{+\infty} \ln(1+ K e^{-t^2}) dt = \! \int_{-\infty}^{+\infty} \! \sum_{n=1}^{\infty} \frac{(-1)^{n-1} K^n }{n}  e^{-n t^2} \! dt \\
		= &\sum_{n=1}^{\infty} \frac{(-1)^{n-1} K^n}{n} \int_{-\infty}^{+\infty} e^{-n t^2} dt = \sum_{n=1}^{\infty} \frac{(-1)^{n-1} K^n}{n} \sqrt{\frac{\pi}{n}} \\
		= &-\sqrt{\pi} \sum_{n=1}^{\infty} \frac{ \left( -K \right)^n}{n^{\frac{3}{2}}} = -\sqrt{\pi} \operatorname{Li}_{3/2} \left( -K \right)
	\end{aligned}
\end{equation}
where $\operatorname{Li}_s \left( z \right) = \sum_{n=1}^{\infty} z^n / n^s$ is the polylogarithm. Simplifying~\eqref{Eq_2D_x0_left_hand_side} with~\eqref{Eq_2D_Li}, the left-hand side of~\eqref{Eq_2D_x0_to_solve} is
\begin{equation}\label{Eq_left_hand_side_final}
	\int_{-\infty}^{+\infty} \! C_{\text{inst}}^{2\text{D}} \! \left( x \right)  dx \! = \! -\frac{d\theta_{3 \text{dB}}}{\ln 2} \! \sqrt{\frac{\pi}{1.2 \ln 10}} \! \operatorname{Li}_{3/2} \! \left( \! - \frac{P_{\max} A_e}{ 4 \pi d^2 N_0} \right) \! .
\end{equation}
Setting~\eqref{Eq_2D_x0_right_hand_side} equal to~\eqref{Eq_left_hand_side_final}, closed-form expression of $\kappa$ can be computed as
\begin{equation}\label{Eq_2D_kappa}
	\begin{aligned}
		\kappa =  &- \frac{5}{2 x_0^5} \frac{d\theta_{3 \text{dB}}}{\ln 2} \sqrt{\frac{\pi}{1.2 \ln 10}} \operatorname{Li}_{3/2} \left( - \frac{P_{\max} A_e}{ 4 \pi d^2 N_0} \right) \\
		& -\frac{10}{3 x_0^4} \log_2 \left( 1 +  \frac{P_{\max} A_e}{ 4 \pi d^2 N_0} \right).
	\end{aligned}
\end{equation}
where $x_0$ is given by~\eqref{Eq_2D_x0}. After obtaining $\kappa$, the closed-form ergodic capacity $C_{\text{erg}}^{2\text{D}}$ in~\eqref{Eq_2D_Cerg_exact} can be approximated as
\begin{equation}\label{Eq_2D_Cerg_infty}
	\begin{aligned}
		C_{\text{erg}}^{2\text{D},+} &= \int_{-x_0}^{x_0} C_{\text{inst}}^{2\text{D},+} \left( x \right) \cdot f_{\hat{x}_u} \left( x \right)  dx \\
		&= \log_2 \left( 1 + \frac{P_{\max} A_e}{ 4 \pi d^2 N_0} \right) \cdot \left[ 1 - 2 Q\left( \frac{x_0}{\sigma_x} \right) \right]  \\
		&- \frac{1.2 \log_2 10}{d^2 \theta_{3 \text{dB}}^2} \cdot \frac{\frac{P_{\max} A_e}{ 4 \pi d^2 N_0}}{\frac{P_{\max} A_e}{ 4 \pi d^2 N_0}  + 1} \cdot  M_2^{2\text{D}} + \kappa  \cdot M_4^{2\text{D}} 
	\end{aligned}
\end{equation}
where $M_2^{2\text{D}}, M_4^{2\text{D}}$ represent the truncated second and fourth moment of $\hat{x}_u$, respectively, which can be computed as
\begin{equation}\label{Eq_2D_M1}
	M_2^{2\text{D}} = \sigma_x^2 \left[ 1 - 2 Q \left( \frac{x_0}{\sigma_x} \right) - \sqrt{\frac{2}{\pi}} \frac{x_0}{\sigma_x} e^{ -\frac{x_0^2}{2 \sigma_x^2}} \right]
\end{equation}
and
\begin{equation}\label{Eq_2D_M2}
	M_4^{2\text{D}} = \sigma_x^4 \left[ 3\left( 1 \! - \! 2 Q \left( \frac{x_0}{\sigma_x} \right) \right) - \sqrt{\frac{2}{\pi}} \frac{x_0}{\sigma_x} \left( \frac{x_0^2}{\sigma_x^2} + 3 \right) e^{ -\frac{x_0^2}{2 \sigma_x^2}}\right]
\end{equation}
where  $Q \left( x \right) = \int_x^{+\infty} 1/ \left( \sqrt{2 \pi} \right) \exp{\left( -t^2/2 \right)} dt$ is the Gaussian $Q$-function, and $x_0$ has been derived in~\eqref{Eq_2D_x0}. As shown in~\eqref{Eq_2D_Cerg_infty}, the ergodic capacity depends on the link distance, beamwidth, and positioning error, which should be jointly considered in beam design to maximize the overall system capacity. It is worth noting that~\eqref{Eq_2D_Cerg_infty} is a two-sided asymptotic approximation, since it remains accurate under both the narrow beam and wide beam conditions in Section~\ref{Two Asymptotic Cases}. To better illustrate the approximation accuracy, a comparison of $C_{\text{inst}}^{2\text{D}} \left( \hat{x}_u \right)$ in~\eqref{Eq_2D_Cinst_x}, $C_{\text{inst}}^{2\text{D},2\text{nd},+} \left( \hat{x}_u \right)$ in~\eqref{Eq_2D_Cinst_2nd_clamp}, and $C_{\text{inst}}^{2\text{D},+} \left( \hat{x}_u \right)$ in~\eqref{Eq_2D_Cinst_infty} is shown in Fig.~\ref{Figure_2D_Different_Cinst}.

\begin{figure}[!t]
	\centering
	\subfloat[Asymptotic case with narrow beamwidth: impulse-like instantaneous capacity and ``flat'' $f_{\hat{x}_u} \left( x \right)$ near the origin.]{
		\includegraphics[width=0.45\textwidth]{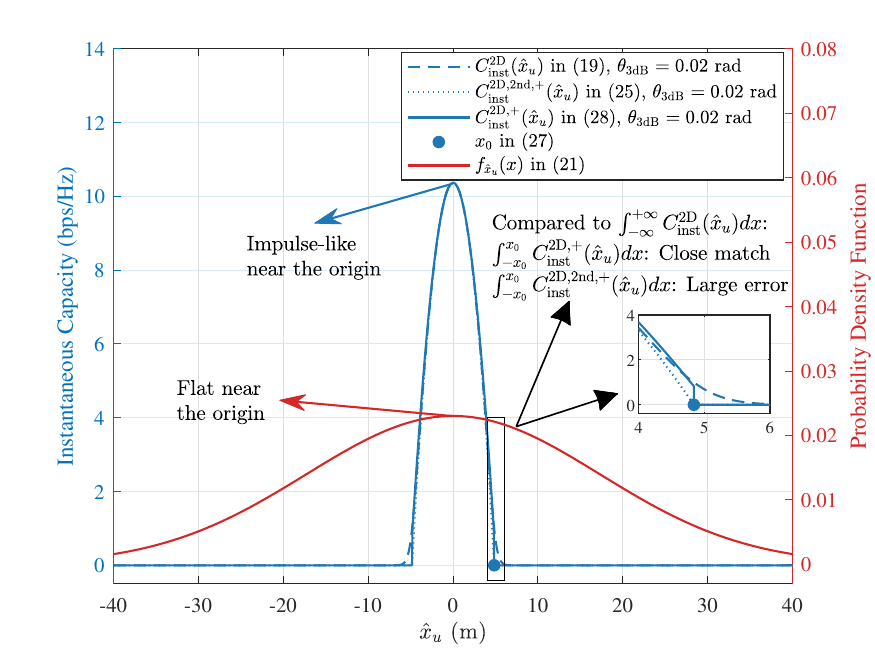}
		\label{Figure_2D_Small_theta}
	}
	\hfill
	\subfloat[Asymptotic case with wide beamwidth: ``flat'' instantaneous capacity and   impulse-like $f_{\hat{x}_u} \left( x \right)$ near the origin.]{
		\includegraphics[width=0.45\textwidth]{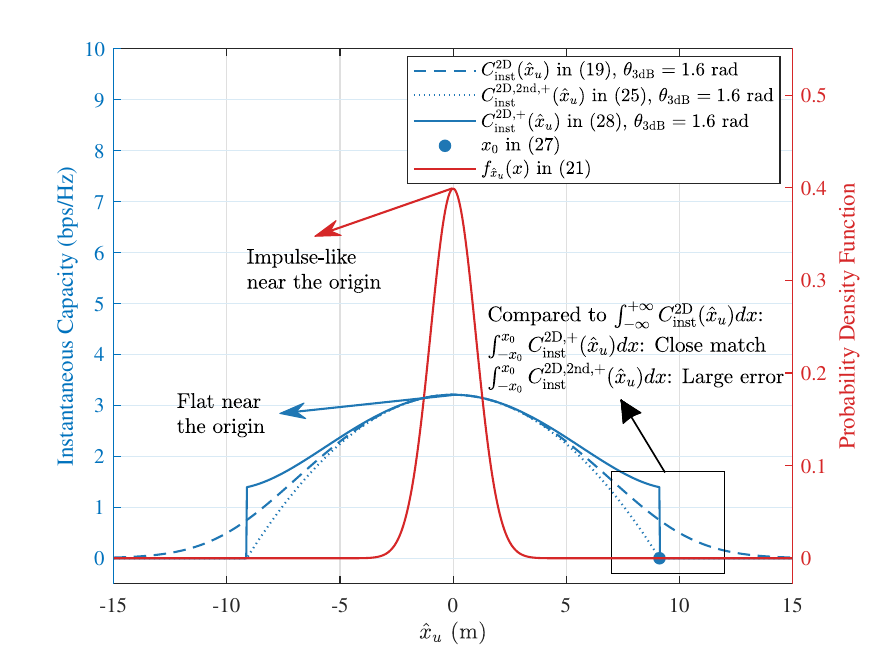}
		\label{Figure_2D_Large_theta}
	}
	\caption{Comparison of two asymptotic cases of $C_{\text{inst}}^{2\text{D}} \left( \hat{x}_u \right)$ in~\eqref{Eq_2D_Cinst_x}, $C_{\text{inst}}^{2\text{D},2\text{nd},+} \left( \hat{x}_u \right)$ in~\eqref{Eq_2D_Cinst_2nd_clamp}, and $C_{\text{inst}}^{2\text{D},+} \left( \hat{x}_u \right)$ in~\eqref{Eq_2D_Cinst_infty} as a function of $\hat{x}_u$.}
	\label{Figure_2D_Different_Cinst}
\end{figure}

Fig.~\eqref{Figure_2D_Small_theta} illustrates the asymptotic case in~\eqref{Eq_2D_Asymptotic_Case1_narrow_beam}. It can be observed that the instantaneous capacity becomes increasingly concentrated around the origin, exhibiting a impulse-like behavior, while $f_{\hat{x}_u} \left( x \right)$ remains nearly constant in the vicinity of the origin. Therefore, when evaluating $C_{\text{erg}}^{2\text{D}}$ in~\eqref{Eq_2D_Cerg_exact}, $f_{\hat{x}_u} \left( x \right)$ can be approximated by $f_{\hat{x}_u} \left( 0 \right)$, as shown in~\eqref{Eq_2D_Asymptotic_Case1_narrow_beam}. Meanwhile, $C_{\text{inst}}^{2\text{D},2\text{nd},+} \left( \hat{x}_u \right)$ in~\eqref{Eq_2D_Cinst_2nd_clamp} exhibits a considerable deviation from the exact expression $C_{\text{inst}}^{2\text{D}} \left( \hat{x}_u \right)$ in~\eqref{Eq_2D_Cinst_x}, while $C_{\text{inst}}^{2\text{D},+} \left( \hat{x}_u \right)$ in~\eqref{Eq_2D_Cinst_infty} significantly improves the approximation accuracy. It is worth noting that, since $x_0$ is the root of the quadratic approximation in~\eqref{Eq_2D_Cinst_2nd} rather than that of the fourth-order approximation in~\eqref{Eq_2D_Cinst_infty}, $C_{\mathrm{inst}}^{2\mathrm{D},+}(\hat{x}_u)$ in~\eqref{Eq_2D_Cinst_infty} has a discontinuity at $x_0$, which is more evident in Fig.~\ref{Figure_2D_Large_theta}. However, this discontinuity does not affect its satisfaction of the two asymptotic cases.

The other asymptotic case of wide beamwidth is shown in Fig.~\eqref{Figure_2D_Large_theta}. The instantaneous capacity remains nearly constant and $f_{\hat{x}_u} \left( x \right)$ behaves like an impulse around the origin. Hence, when evaluating $C_{\text{erg}}^{2\text{D}}$ in~\eqref{Eq_2D_Cerg_exact}, $C_{\text{inst}}^{2\text{D}} \left( \hat{x}_u \right)$ can be approximated by $C_{\text{inst}}^{2\text{D}} \left( 0 \right)$. It is worth noting that, when the beamwidth is sufficiently large, although $C_{\text{inst}}^{2\text{D},2\text{nd},+} \left( \hat{x}_u \right)$ in~\eqref{Eq_2D_Cinst_2nd_clamp} exhibits noticeable deviation from $C_{\text{inst}}^{2\text{D}} \left( \hat{x}_u \right)$ in~\eqref{Eq_2D_Cinst_x} away from the origin, this deviation has a negligible impact because the positioning error is highly concentrated around the origin and the beam can cover most of the positioning error region.

\section{3D Asymptotic Capacity Expression}
\subsection{3D Ergodic Capacity in an Integral Form}
Similar to 2D cases, the instantaneous channel capacity of 3D positioning-assisted beamforming systems can be calculated according to~\eqref{Eq_Pr},~\eqref{Eq_Gd} and~\eqref{Eq_3D_Gtheta_phi} as
\begin{equation}\label{Eq_3D_Cinst}
	\begin{aligned}
		C_{\text{inst}}^{3\text{D}} &= \log_2 \! \left( 1 + \frac{P_r}{N_0} \right) \\
		&= \log_2 \left( 1 + \frac{P_{\max} A_e}{ 4 \pi d^2 N_0}\cdot  10^{-1.2 \left[ \theta, \phi \right] \mathbf{A} \left[ \theta, \phi \right]^T} \right) .
	\end{aligned}
\end{equation}
As error is small compared to the link distance, we have $\hat{x}_u / d \rightarrow 0$, $\hat{y}_u / d \rightarrow 1$ and $\hat{z}_u / d \rightarrow 0$. According to~\eqref{Eq_3D_theta} and~\eqref{Eq_3D_phi}, $\theta$ and $\phi$ can be approximated by Taylor expansion as
\begin{equation}\label{Eq_3D_theta_Taylor}
	\theta = \arctan \frac{\hat{x}_u}{\hat{y}_u} = \frac{\hat{x}_u}{d} + o \left( || \left( \hat{x}_u, \left( \hat{y}_u - d \right), \hat{z}_u \right) || \right)
\end{equation}
and
\begin{equation}\label{Eq_3D_phi_Taylor}
	\phi = \arctan \frac{\hat{z}_u}{\sqrt{\hat{x}_u^2 + \hat{y}_u^2}} = \frac{\hat{z}_u}{d} + o \left( || \left( \hat{x}_u, \left( \hat{y}_u - d \right), \hat{z}_u \right) || \right)
\end{equation}
when $\left( \hat{x}_u, \hat{y}_u, \hat{z}_u  \right) \to \left( 0, d, 0 \right)$. Substituting~\eqref{Eq_3D_theta_Taylor} and~\eqref{Eq_3D_phi_Taylor} into~\eqref{Eq_3D_Cinst}, we can approximate $C_{\text{inst}}^{3\text{D}}$ as
\begin{equation}\label{Eq_3D_Cinst_xz}
	C_{\text{inst}}^{3\text{D}} \! \left( \hat{x}_u,\hat{z}_u \right)  \approx \log_2 \left( 1 \! + \! \frac{P_{\max} A_e}{ 4 \pi d^2 N_0}  \cdot 10^{-\frac{1.2}{d^2} \left[ \hat{x}_u,\hat{z}_u \right] \mathbf{A} \left[ \hat{x}_u,\hat{z}_u \right]^T} \right) \! .
\end{equation}
Combining~\eqref{Eq_3D_PDF} and~\eqref{Eq_3D_Cinst_xz}, the ergodic capacity can be computed by marginalizing $\hat{y}_u$ as
\begin{equation}\label{Eq_3D_Cerg_exact}
	\begin{aligned}
		C_{\text{erg}}^{3\text{D}} &= \iiint_{-\infty}^{+\infty} C_{\text{inst}}^{3\text{D}} \left( x, z \right) \cdot f_{\hat{\mathbf{p}}_u^{3\text{D}}} \left( \hat{\mathbf{p}}_u^{3\text{D}} \right) dxdydz \\
		&=\iint_{-\infty}^{+\infty}C_{\text{inst}}^{3\text{D}} \left( x, z \right) \cdot f_{\hat{x}_u, \hat{z}_u} \left( x, z \right) dx dz
	\end{aligned}
\end{equation}
where
\begin{equation}\label{Eq_3D_fxz}
	f_{\hat{x}_u, \hat{z}_u} \left( x, z \right) = \frac{1}{2 \pi \!\left| \mathbf{\Sigma}^{'2\text{D}} \right|^{\frac{1}{2}}} \exp \Big( -\frac{ \left[ x, z \right] \left( \mathbf{\Sigma}^{'2\text{D}} \right)^{-1} \left[ x, z \right] ^T }{2} \Big)
\end{equation}
denotes the joint PDF of $\left[ \hat{x}_u, \hat{z}_u \right]$ and 
\begin{equation}\label{Eq_3D_Cov_Matrix_xz}
	\mathbf{\Sigma}^{'2\text{D}} =  
	\begin{bmatrix}
		\mathbf{\Sigma}^{3\text{D}} \left( 1,1 \right) &\mathbf{\Sigma}^{3\text{D}} \left( 1,3 \right) \\ \mathbf{\Sigma}^{3\text{D}} \left( 3,1 \right)  & \mathbf{\Sigma}^{3\text{D}} \left( 3,3 \right)
	\end{bmatrix} 
\end{equation}
represents their covariance matrix where $\mathbf{\Sigma}^{3\text{D}} \left( m,n \right)$ is the $\left( m, n \right)$th entry of $\mathbf{\Sigma}^{3\text{D}}$.

\subsection{Approximations to $C_{\rm{inst}}^{3\rm{D}}$}
Similar to the 2D asymptotic cases in Section~\ref{Two Asymptotic Cases}, eq.~\eqref{Eq_3D_Cinst_xz} can also be simplified into two asymptotic cases. Analogously, we propose the polynomial expansion technique to the exact 3D instantaneous capacity $C_{\text{inst}}^{3\text{D}} \left( \hat{x}_u,\hat{z}_u \right)$ in~\eqref{Eq_3D_Cinst_xz} as
\begin{equation}\label{Eq_3D_Cinst_infty}
	\begin{aligned}
		&C_{\text{inst}}^{3\text{D},+} \! \left( \hat{x}_u, \hat{z}_u \right) \! = \\
		&\!\begin{cases}
			&\!\!\!\!\!\!\log_2 \! \left( \! 1 \! + \! \frac{P_{\max} A_e}{ 4 \pi d^2 \! N_0} \! \right) \! - \! \frac{1.2 \log_2 \! 10}{d^2} \! \cdot \! \frac{\frac{P_{\max} A_e}{ 4 \pi d^2 N_0}}{\frac{P_{\max} A_e}{ 4 \pi d^2 N_0} \! + \! 1} \! \left[ \hat{x}_u,\! \hat{z}_u \right] \! \mathbf{A} \! \left[ \hat{x}_u, \!\hat{z}_u \right]^T \\
			&+\, \eta \left( \left[ \hat{x}_u, \hat{z}_u \right] \mathbf{A} \left[ \hat{x}_u, \hat{z}_u \right]^T \right)^2,  \left[ \hat{x}_u, \hat{z}_u \right] \mathbf{A} \left[ \hat{x}_u, \hat{z}_u \right]^T \leq R_0^2 \\
			&\!\!\!\!\!\! 0, \qquad \qquad \qquad \qquad \qquad \quad \;\,\,\,\, \left[ \hat{x}_u, \hat{z}_u \right] \mathbf{A} \left[ \hat{x}_u, \hat{z}_u \right]^T > R_0^2
		\end{cases}
	\end{aligned}
\end{equation}
where the first two terms in the non-zero branch are obtained from the Taylor expansion of $C_{\text{inst}}^{3\text{D}} \! \left( \hat{x}_u,\hat{z}_u \right)$ in~\eqref{Eq_3D_Cinst_xz}; $\eta$ serves the same role as the $\kappa$ in 2D scenarios and is chosen to ensure
\begin{equation}\label{Eq_3D_Cinst=Cinst_infty}
	\iint_{-\infty}^{+\infty} C_{\text{inst}}^{3\text{D}} \left( x, z \right) dx dz =  \iint_{\mathcal{D}_0}  C_{\text{inst}}^{3\text{D},+} \left( x, z \right) dx dz 
\end{equation}
where $\mathcal{D}_0$ satisfies
\begin{equation}\label{Eq_3D_D0}
	\mathcal{D}_0 = \left\{ \left[ x, z \right] \middle| \left[ x, z \right] \mathbf{A} \left[ x, z \right]^T \leq R_0^2 \right\}
\end{equation}
and $R_0^2$ is chosen as the root of
\begin{equation}\label{Eq_3D_eta_to_solve}
	\log_2 \! \left(  1 \! + \! \frac{P_{\max} A_e}{ 4 \pi d^2 \! N_0} \right) - \frac{1.2 \log_2  10}{d^2} \! \cdot \! \frac{\frac{P_{\max} A_e}{ 4 \pi d^2 N_0}}{\frac{P_{\max} A_e}{ 4 \pi d^2 N_0}  +  1} \cdot  R_0^2 =  0,
\end{equation}
which can be calculated as
\begin{equation}\label{Eq_3D_R0^2}
	R_0^2 = \log_2  \left(  1  + \frac{P_{\max} A_e}{ 4 \pi d^2  N_0} \right) \cdot   \frac{d^2}{1.2 \log_2  10}  \cdot  \frac{\frac{P_{\max} A_e}{ 4 \pi d^2 N_0} +  1}{\frac{P_{\max} A_e}{ 4 \pi d^2 N_0}}.
\end{equation}

\subsection{Calculation of $\eta$ and Asymptotic Approximation to $C_{\rm{erg}}^{3\rm{D}}$}
Substituting $C_{\text{inst}}^{3\text{D}} \left( \hat{x}_u,\hat{z}_u \right)$ in~\eqref{Eq_3D_Cinst_xz} and $C_{\text{inst}}^{3\text{D},+} \left( \hat{x}_u,\hat{z}_u \right)$ in~\eqref{Eq_3D_Cinst_infty} into~\eqref{Eq_3D_Cinst=Cinst_infty}, we can rewrite~\eqref{Eq_3D_Cinst=Cinst_infty} as
\begin{equation}\label{Eq_3D_R0_to_solve}
	\begin{aligned}
		&\underbrace{\iint_{-\infty}^{+\infty} \log_2 \left( \! 1 \! + \! \frac{P_{\max} A_e}{ 4 \pi d^2 N_0} \! \cdot \! 10^{-\frac{1.2}{d^2} \left[ x, z \right] \mathbf{A} \left[ x, z \right]^T} \right) dx dz}_{I_L} \\
		&=  \underbrace{
			\begin{aligned}[t]
				&\iint_{\mathcal{D}_0} \left[ \log_2 \left( 1  +  \frac{P_{\max} A_e}{ 4 \pi d^2 N_0} \right) \right. \\
		& \qquad \quad \!\! - \frac{1.2 \log_2 10}{d^2} \cdot \frac{\frac{P_{\max} A_e}{ 4 \pi d^2 N_0}}{\frac{P_{\max} A_e}{ 4 \pi d^2 N_0} + 1} \cdot \left[ x, z \right] \mathbf{A} \left[ x, z \right]^T \\
		&\qquad \quad  \!\! + \left. \eta \left( \left[ x, z \right] \mathbf{A} \left[ x, z \right]^T \right)^2 \right] dx dz
		\end{aligned}
	}_{I_R}
	\end{aligned}
\end{equation}
where $I_R$ and $I_L$ can be evaluated separately. To calculate $I_R$, we apply the eigenvalue decomposition $\mathbf{A} = \mathbf{U} \mathbf{Y} \mathbf{U}^T$ and change variables to $\mathbf{t} = \sqrt{\mathbf{Y}} \mathbf{U}^T \left[ x, z \right]^T$. After simplification, $I_R$ can be expressed as
\begin{equation}\label{Eq_3D_IR_final}
	\begin{aligned}
		I_R &= \frac{2 \pi}{\left| \mathbf{A} \right|^{\frac{1}{2}}} \cdot \left[ \frac{R_0^2}{2} \cdot \log_2 \left( 1  +  \frac{P_{\max} A_e}{ 4 \pi d^2 N_0} \right) \right. \\
		& \qquad  \qquad - \left. \frac{R_0^4}{4} \cdot \frac{1.2 \log_2 \! 10}{d^2} \cdot \frac{\frac{P_{\max} A_e}{ 4 \pi d^2 N_0}}{\frac{P_{\max} A_e}{ 4 \pi d^2 N_0} + 1} + \frac{R_0^6}{6} \cdot \eta \right].
	\end{aligned}
\end{equation}
$I_L$ can be equivalently transformed as
\begin{equation}\label{Eq_3D_IL_exp}
	\begin{aligned}
		I_L  =  \iint_{-\infty}^{+\infty} & \log_2 \left( 1 + \frac{P_{\max} A_e}{ 4 \pi d^2  N_0} \right. \\
		& \times \left. \exp  \left( -\frac{1.2 \ln 10}{d^2} \left[ x, z \right] \mathbf{A} \left[ x, z \right]^T \right) \right)  dx dz.
	\end{aligned}
\end{equation}
To calculate~\eqref{Eq_3D_IL_exp}, we change variables in~\eqref{Eq_3D_IL_exp} to $\mathbf{v} = \sqrt{\frac{1.2 \ln 10}{d^2}} \left[ \hat{x}_u, \hat{z}_u \right]^T$ and after simplification, eq.~\eqref{Eq_3D_IL_exp} can be expressed as
\begin{equation}\label{Eq_3D_IL_final}
	I_L = -\frac{\pi d^2}{ \left| \mathbf{A} \right|^{\frac{1}{2}} \cdot  \ln 2 \cdot 1.2 \ln 10}  \cdot \operatorname{Li}_2 \left( -\frac{P_{\max} A_e}{ 4 \pi d^2 N_0 } \right).
\end{equation}
Setting $I_R$ in~\eqref{Eq_3D_IR_final} equal to $I_L$ in~\eqref{Eq_3D_IL_final} and combining with~\eqref{Eq_3D_R0^2}, closed-form expression of $\eta$ can be obtained as
\begin{equation}\label{Eq_3D_eta}
	\begin{aligned}
		\eta = &-\frac{3 d^2}{R_0^6 \cdot \ln 2 \cdot 1.2 \ln 10} \cdot \operatorname{Li}_2 \left( -\frac{P_{\max} A_e}{ 4 \pi d^2 N_0 } \right) \\
		&- \frac{3}{2 R_0^4} \log_2 \left( 1  +  \frac{P_{\max} A_e}{ 4 \pi d^2 N_0} \right).
	\end{aligned}
\end{equation}
Substituting $C_{\text{inst}}^{3\text{D},+} \left( \hat{x}_u, \hat{z}_u \right)$ in~\eqref{Eq_3D_Cinst_infty} into~\eqref{Eq_3D_Cerg_exact} in place of $C_{\text{inst}}^{3\text{D}} \left( x, z \right)$ and combining with~\eqref{Eq_3D_eta}, we can approximate the closed-form expression of $C_{\text{erg}}^{3\text{D}}$ as
\begin{equation}\label{Eq_3D_Cerg_infty}
	\begin{aligned}
		C_{\text{erg}}^{3\text{D},+} &= \iint_{\mathcal{D}_0} C_{\text{inst}}^{3\text{D},+} \left( x, z \right) \cdot f_{\hat{x}_u, \hat{z}_u} \left( x, z \right) dx dz \\
		&=  \log_2 \left( 1  +  \frac{P_{\max} A_e}{ 4 \pi d^2 N_0} \right) \cdot M_0^{3\text{D}} \\
		&- \frac{1.2 \log_2 10}{d^2} \cdot \frac{\frac{P_{\max} A_e}{ 4 \pi d^2 N_0}}{\frac{P_{\max} A_e}{ 4 \pi d^2 N_0} + 1} \cdot M_2^{3\text{D}} + \eta \cdot M_4^{3\text{D}}
	\end{aligned}
\end{equation}
where $M_0^{3\text{D}}$, $M_2^{3\text{D}}$ and $M_4^{3\text{D}}$ represent the truncated zero, second and fourth moment of Gaussian distribution, whose closed-form expressions are derived in Appendix~\ref{Appendix_M0^3D},~\ref{Appendix_M2^3D} and~\ref{Appendix_M4^3D}, respectively. As indicated by~\eqref{Eq_3D_Cerg_infty}, the 3D ergodic channel capacity is governed by the link distance, beamwidth, and positioning error. This behavior is consistent with that observed in the 2D case, suggesting that these parameters should be jointly optimized to maximize the capacity.

\section{Beam Pattern Optimization}\label{Beam Pattern Optimization}
\subsection{2D Optimal Beam Pattern}
Since~\eqref{Eq_2D_Cerg_infty} involves $Q(\cdot)$, $\operatorname{Li}(\cdot)$ functions, direct optimization is difficult. We therefore seek an asymptotic expression for~\eqref{Eq_2D_Cerg_infty} to facilitate the optimization. To simplify $Q \left( x_0 / \sigma_x \right)$ in~\eqref{Eq_2D_Cerg_infty}, we consider the asymptotic Case 2 in~\ref{Two Asymptotic Cases}, where $\theta_{3 \text{dB}} d / \sigma_x \to +\infty$. According to~\eqref{Eq_2D_x0}, we can calculate that
\begin{equation}\label{Eq_2D_x0/sigmax}
	\frac{x_0}{\sigma_x} \! = \! \frac{d \theta_{3 \text{dB}}}{\sigma_x} \! \sqrt{\frac{\log_2 \! \left( 1 + \frac{P_{\max} A_e}{ 4 \pi d^2 \! N_0} \right) }{1.2 \log_2 10} \! \cdot \! \frac{\frac{P_{\max} A_e}{ 4 \pi d^2 N_0} + 1}{\frac{P_{\max} A_e}{ 4 \pi d^2 N_0}} } \to \! +\infty,
\end{equation}
which results in $Q \left( x_0 / \sigma_x \right) \to 0$, $M_2^{2\text{D}} \to \sigma_x^2$, $M_4^{2\text{D}} \to 3\sigma_x^4$, and the integral region in~\eqref{Eq_2D_Cerg_infty} can be extended to $\left( -\infty, +\infty \right)$. Moreover, $\kappa$ in~\eqref{Eq_2D_kappa} tends to zero as $x_0/\sigma_x$ increases. In this way, the ergodic capacity in~\eqref{Eq_2D_Cerg_infty} can be asymptotically approximated by
\begin{equation}\label{Eq_2D_Cerg_with_Pmax}
	\begin{aligned}
		C_{\text{erg}}^{2\text{D},+} & \approx \int_{-\infty}^{+\infty} \left[  \log_2  \left( 1 +  \frac{P_{\max} A_e}{ 4 \pi d^2 N_0} \right) \right. \\
		&\qquad \qquad \left. - \frac{1.2 \log_2 10 \, x^2}{d^2 \theta_{3 \text{dB}}^2} \cdot \frac{\frac{P_{\max} A_e}{ 4 \pi d^2 N_0}}{\frac{P_{\max} A_e}{ 4 \pi d^2 N_0}  + 1} \right] \cdot f_{\hat{x}_u} \left( x \right)  dx \\
		&= \log_2  \left( \! 1 \! + \! \frac{P_{\max} A_e}{ 4 \pi d^2 \! N_0}\! \right) \! - \! \frac{1.2 \log_2 \! 10 \, \sigma_x^2}{d^2 \theta_{3 \text{dB}}^2} \! \cdot \! \frac{\frac{P_{\max} A_e}{ 4 \pi d^2 \! N_0}}{\frac{P_{\max} A_e}{ 4 \pi d^2 \! N_0} \! + \! 1} .
	\end{aligned}
\end{equation}
In the high-SNR regime, $\frac{P_{\max} A_e}{ 4 \pi d^2 \! N_0} \gg 1$. Substituting~\eqref{Eq_2D_Pmax} into~\eqref{Eq_2D_Cerg_with_Pmax} to eliminate $P_{\max}$, the channel capacity expression reduces to
\begin{equation}\label{Eq_2D_Cerg_to_be_opt}
	C_{\text{erg}}^{2\text{D},+} \approx  \log_2 \left( \frac{A_e P_t}{ \theta_{3\text{dB}} 4 \pi d^2 N_0} \sqrt{\frac{1.2 \ln 10}{\pi}} \right) - \frac{1.2 \log_2 10 \, \sigma_x^2}{d^2 \theta_{3 \text{dB}}^2}.
\end{equation}
Taking the derivative of $\theta_{3\text{dB}}$ in~\eqref{Eq_2D_Cerg_to_be_opt}, we have
\begin{equation}\label{Eq_2D_dtheta}
	\frac{d C_{\text{erg}}^{2\text{D},+}}{d \theta_{3\text{dB}}} = - \frac{1}{\theta_{3\text{dB}} \ln 2} + \frac{2.4 \log_2 10 \, \sigma_x^2}{d^2 \theta_{3\text{dB}}^3}.
\end{equation}
By setting~\eqref{Eq_2D_dtheta} to zero, we can calculate the optimal beamwidth as
\begin{equation}\label{Eq_2D_Optimal_theta}
	\theta_{3 \text{dB}}^\ast = \frac{1}{d} \sqrt{2.4 \ln 10\mathbf{\Sigma}^{2\text{D}} \left( 1,1 \right)}
\end{equation}
where $\mathbf{\Sigma}^{2\text{D}} \left( m, n \right)$ is the $\left(m, n \right)$th entry of $\mathbf{\Sigma}^{2\text{D}}$. The results reveal that the optimal beam is determined by the link distance and the variance along the $x$-axis, while remaining invariant to transmit power variations. This is because changes in the transmit power only scale the maximum channel capacity but do not change the optimal $\theta_{3 \text{dB}}^\ast$.

\subsection{3D Optimal Beam Pattern}
Similar to the 2D cases, under the asymptotic case when the beam pattern much larger than the ratio of the positioning error to the link distance $d$, $\mathcal{D}_0$ in~\eqref{Eq_3D_Cerg_infty} expands to cover the entire  $\mathbb{R}^2$. $M_0^{3\text{D}}$ and $M_2^{3\text{D}}$ in~\eqref{Eq_3D_Cerg_infty} can be approximated as
\begin{equation}\label{Eq_3D_M0_approximate}
	M_0^{3\text{D}} \approx \iint_{-\infty}^{+\infty} f_{\hat{x}_u, \hat{z}_u} \left( x, z \right) dxdz = 1
\end{equation}
and
\begin{equation}\label{Eq_3D_M2_approximate}
	\begin{aligned}
		M_2^{3\text{D}} &\approx \iint_{-\infty}^{+\infty} \left[ x, z \right] \mathbf{A} \left[ x, z \right]^T f_{\hat{x}_u, \hat{z}_u} \left( x, z \right) dxdz \\
		&= \mathbb{E} \left[ \left[ x, z \right] \mathbf{A} \left[ x, z \right]^T \right] = \mathbb{E} \left[ \operatorname{tr} \left( \mathbf{A} \left[ x, z \right]^T \left[ x, z \right] \right) \right] \\
		&= \operatorname{tr} \left( \mathbf{A} \mathbb{E} \left[ \left[ x, z \right]^T  \left[ x, z \right] \right] \right) = \operatorname{tr} \left( \mathbf{A} \mathbf{\Sigma}^{'2\text{D}} \right).
	\end{aligned}
\end{equation}
Similar to $\kappa$ in the 2D cases, we also have $\eta \! \to 0$. Substituting~\eqref{Eq_3D_M0_approximate},~\eqref{Eq_3D_M2_approximate} and $\eta \! = \! 0$ into~\eqref{Eq_3D_Cerg_infty}, we can approximate that 
\begin{equation}\label{Eq_3D_Cerg_with_Pmax}
		\begin{aligned}
		C_{\text{erg}}^{3\text{D},+} &\approx \log_2 \left( 1  +  \frac{P_{\max} A_e}{ 4 \pi d^2 N_0} \right)  \\
		&- \frac{1.2 \log_2 10}{d^2} \cdot \frac{\frac{P_{\max} A_e}{ 4 \pi d^2 N_0}}{\frac{P_{\max} A_e}{ 4 \pi d^2 N_0} + 1} \cdot \operatorname{tr} \left( \mathbf{A} \mathbf{\Sigma}^{'2\text{D}} \right) .
	\end{aligned}
\end{equation}
In the high-SNR regime, $\frac{P_{\max} A_e}{ 4 \pi d^2 \! N_0} \gg 1$. Substituting~\eqref{Eq_3D_Pmax} into~\eqref{Eq_3D_Cerg_with_Pmax} to eliminate $P_{\max}$, the 3D capacity reduces to
\begin{equation}\label{Eq_3D_Cerg_to_be_opt}
	\begin{aligned}
		C_{\text{erg}}^{3\text{D},+} &= \log_2 \left( \frac{1.2 \ln 10 A_e P_t}{ \pi \left(4 \pi d \right)^2  N_0 } \right) +  \frac{1}{2}  \log_2 \left|  \mathbf{A} \right|   \\
		& -  \frac{1.2 \log_2 10}{d^2}  \operatorname{tr} \left(  \mathbf{A} \mathbf{\Sigma}^{'2\text{D}}  \right).
	\end{aligned}
\end{equation}

Taking the gradient of~\eqref{Eq_3D_Cerg_to_be_opt} with respect to $\mathbf{A}$ and setting it equal to zero as
\begin{equation}\label{Eq_3D_Cerg_gradient}
	\begin{aligned}
		&\nabla_\mathbf{A} \left( \frac{1}{2} \log_2 \left|  \mathbf{A} \right| - \frac{1.2 \log_2 10}{d^2} \operatorname{tr} \left( \mathbf{A} \mathbf{\Sigma}^{'2\text{D}} \right)  \right) \\
		&= \frac{1}{2 \ln 2} \mathbf{A}^{-1} - \frac{1.2 \log_2 10}{d^2} \mathbf{\Sigma}^{'2\text{D}} = 0
	\end{aligned}
\end{equation}
and the optimal $\mathbf{A}^\ast$ can be computed as
\begin{equation}\label{Eq_3D_optimal_A}
	\mathbf{A}^\ast = \frac{d^2}{2.4 \ln 10} \left(   \mathbf{\Sigma}^{'2\text{D}} \right)^{-1}.
\end{equation}
The result in~\eqref{Eq_3D_optimal_A} is insightful, since the channel capacity is optimized when the beam pattern and the error distribution have the same spatial characteristics, which enables the beam to cover the user distribution more effectively. Substituting~\eqref{Eq_3D_A} and~\eqref{Eq_3D_Cov_Matrix_xz} into~\eqref{Eq_3D_optimal_A}, the optimal beam pattern is calculated as
\begin{equation}\label{Eq_3D_optimal_beamwidth}
	\begin{aligned}
		\theta_{3\text{dB}}^\ast &=  \frac{1}{d} \left( 2.4 \ln 10 \cdot \left( \mathbf{\Sigma}^{3\text{D}} \left( 1,1 \right) - \frac{\left( \mathbf{\Sigma}^{3\text{D}} \left( 1,3 \right)  \right)^2}{\mathbf{\Sigma}^{3\text{D}} \left( 3,3 \right) } \right)  \right)^{\frac{1}{2}}, \\
		\phi_{3\text{dB}}^\ast &= \frac{1}{d} \left( 2.4 \ln 10 \cdot \left( \mathbf{\Sigma}^{3\text{D}} \left( 3,3 \right)  - \frac{\left( \mathbf{\Sigma}^{3\text{D}} \left( 1,3 \right)  \right)^2}{\mathbf{\Sigma}^{3\text{D}} \left( 1,1 \right) } \right) \right)^{\frac{1}{2}},  \\
		m^\ast &= -\frac{d^2 \mathbf{\Sigma}^{3\text{D}} \left( 1,3 \right) }{2.4 \ln 10 \cdot \left( \mathbf{\Sigma}^{3\text{D}} \left( 1,1 \right)  \mathbf{\Sigma}^{3\text{D}} \left( 3,3 \right)  - \left( \mathbf{\Sigma}^{3\text{D}} \left( 1,3 \right)  \right)^2 \right)}.
	\end{aligned}
\end{equation}
Notably, $m$ depends on the beam rotation angle $\psi$, defined as the angle from the positive $x$-axis to the beam's major axis. Since $\mathbf{A}$ can be eigendecomposed as
\begin{equation}\label{Eq_3D_eigendecompose_A}
	\mathbf{A} \! = \! \mathbf{U} \mathbf{Y} \mathbf{U}^T \! = \!
	\begin{bmatrix}
		\cos\psi & -\sin\psi \\
		\sin\psi & \cos\psi
	\end{bmatrix}
	\begin{bmatrix}
		\zeta_1 & 0 \\
		0 & \zeta_2
	\end{bmatrix}
	\begin{bmatrix}
		\cos\psi & \sin\psi \\
		-\sin\psi & \cos\psi
	\end{bmatrix} \! ,
\end{equation}
the optimal $\psi^\ast$ can be evaluated by combining~\eqref{Eq_3D_A},~\eqref{Eq_3D_optimal_beamwidth} and~\eqref{Eq_3D_eigendecompose_A} as
\begin{equation}\label{Eq_3D_optimal_psi}
	\psi^\ast = \frac{1}{2} \mathrm{atan2} \left( 2 \mathbf{\Sigma}^{3\text{D}} \left( 1,3 \right), \mathbf{\Sigma}^{3\text{D}} \left( 1,1 \right) - \mathbf{\Sigma}^{3\text{D}} \left( 3,3 \right) \right).
\end{equation}
Equation~\eqref{Eq_3D_optimal_beamwidth} indicates that the beam pattern that optimizes the channel capacity depends on both the link distance and the user position uncertainty. Longer link distance favor narrower beams to increase signal strength, while larger positioning errors require wider beams to effectively cover the potential user distribution.

\begin{figure}
	\centering
	\includegraphics[width=0.45\textwidth]{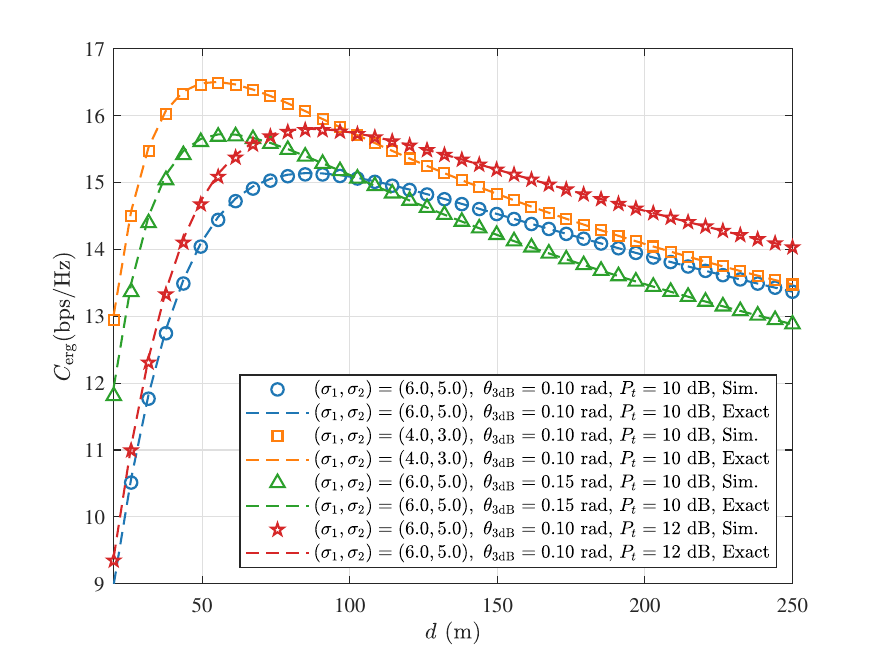}\\
	\caption{Capacity as a function of the link distance $d$ for 2D beamforming systems. $A_e = 1$~cm$^2$; $\varphi = \pi / 3$~rad.}\label{Figure_2D_d}
\end{figure}

\begin{figure}
	\centering
	\includegraphics[width=0.45\textwidth]{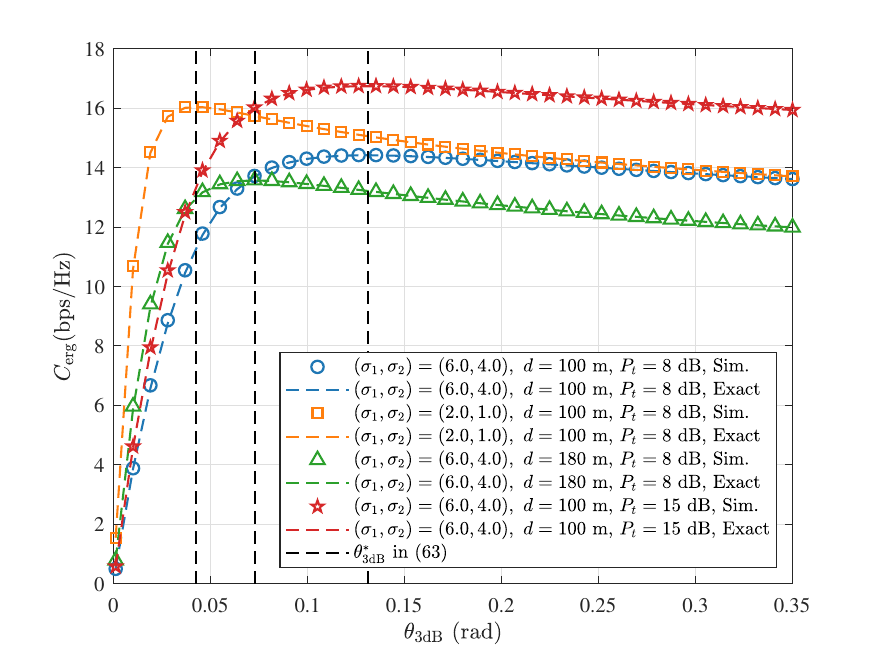}\\
	\caption{Capacity as a function of the 3-dB beamwidth $\theta_{3\text{dB}}$ for 2D beamforming systems. $A_e = 1$~cm$^2$; $\varphi = \pi / 6$~rad. The optimal $\theta_{3\text{dB}}^\ast$ of the blue and red curves overlap.}\label{Figure_2D_theta}
\end{figure}

\section{Numerical Results}\label{Numerical Results}
In this section, we show numerical results for the derived capacity expressions of both 2D and 3D positioning-assisted beamforming systems under various parameter settings. In all figures, ``Sim.'' denotes Monte Carlo simulation results, and ``Exact'' denotes the expression in~\eqref{Eq_2D_Cerg_infty} for the 2D case and~\eqref{Eq_3D_Cerg_infty} for the 3D case.

Figure~\ref{Figure_2D_d} illustrates the relationship between channel capacity and link distance $d$ under varying positioning errors, $\theta_{3\text{dB}}$, and transmit power. The theoretical and simulated capacity curves show an excellent agreement, which confirms the validity of analytical derivations. The capacity shows a non-monotonic dependence on distance, first increasing and then decreasing with larger $d$. This variation is attributed to the balance between beam coverage and path loss. When the distance is small, the beam cannot fully cover the user region, while for large distance the received power decreases rapidly due to propagation attenuation. As a result, the maximum capacity is achieved at a moderate distance. In addition, smaller positioning errors and higher transmit powers lead to larger capacity, while variations in beamwidth influence the results through their effect on both coverage and interference.

\begin{figure}
	\centering
	\includegraphics[width=0.45\textwidth]{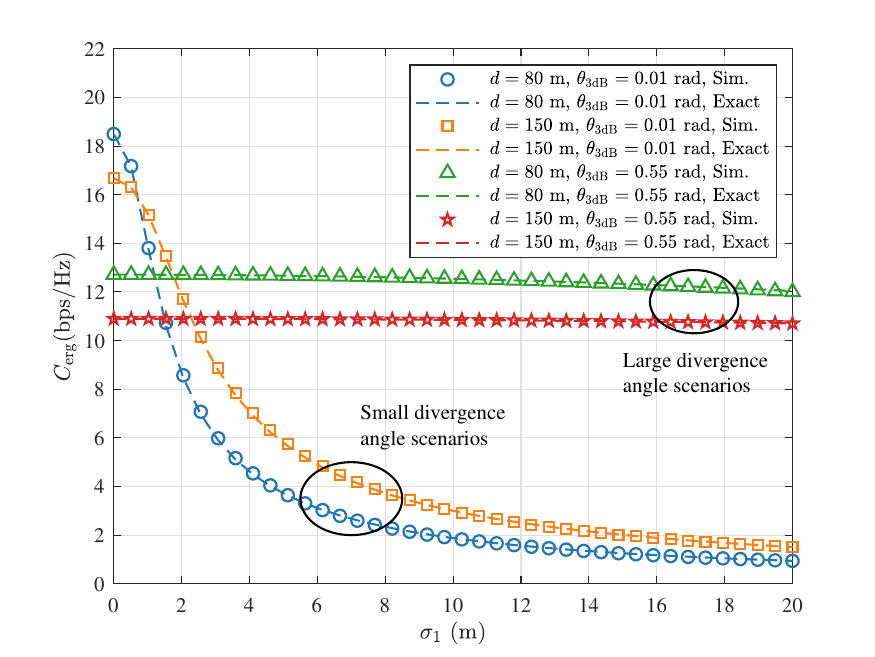}\\
	\caption{Capacity as a function of the error standard deviation $\sigma_1$ along the major axis of the PDF contours for 2D beamforming systems. $P_t = 5$~dB; $A_e = 1$~cm$^2$; $\varphi = \pi / 4$~rad; $\sigma_2 = 0.8 \, \sigma_1$.}\label{Figure_2D_sigma1}
\end{figure}

Figure~\ref{Figure_2D_theta} characterizes the channel capacity versus $\theta_{3\text{dB}}$ under various positioning errors, link distance, and transmit power. The capacity reaches its maximum at an intermediate beamwidth, decaying toward both wider and narrower beamwidth. The maximum capacity computed by~\eqref{Eq_2D_Optimal_theta} are indicated in the figure by vertical dashed lines and align with the observed trend. Systems with lower positioning errors or higher transmit power consistently achieve enhanced capacity across all beamwidth values. In contrast, the effect of link distance on capacity changes with $\theta_{3\text{dB}}$.

Figure~\ref{Figure_2D_sigma1} presents the channel capacity as a function of the error standard deviation $\sigma_1$ along the major axis of the PDF contours under varying link distance and $\theta_{3\text{dB}}$. As $\sigma_1$ increases, the channel capacity decreases due to the growing mismatch between the beam coverage and the user location uncertainty. For small divergence angles with $\theta_{3\text{dB}} = 0.01$~rad, the highly concentrated beam renders the system more sensitive to positioning errors, resulting in rapid capacity degradation at large $\sigma_1$. In contrast, for large divergence angles with $\theta_{3\text{dB}} = 0.45$~rad, the wider beam coverage alleviates the impact of positioning errors and leads to a slower capacity degradation. It is further observed that the orange curve associated with a larger $d$ achieves higher channel capacity than the blue curve with a shorter $d$. This behavior agrees with~\eqref{Eq_2D_Optimal_theta}, where the optimal beamwidth is jointly determined by the link distance and positioning error.

It is worth noting that the above figures are obtained by using~\eqref{Eq_2D_Cinst_infty} rather than~\eqref{Eq_2D_Cinst_x} to compute the capacity in~\eqref{Eq_2D_Cerg_infty}. Although the equivalence between~\eqref{Eq_2D_Cinst_x} and~\eqref{Eq_2D_Cinst_infty} is established only under the two asymptotic conditions in Section~\ref{Two Asymptotic Cases}, the agreement between the Sim. and Exact curves demonstrates that the approximation remains sufficiently accurate in other regimes.

Figure~\ref{Figure_3D_d} illustrates the impact of the link distance $d$ on the capacity under different positioning errors, $\theta_{3\text{dB}}$ and transmit power. The variations of the curves in 3D scenarios follow a trend similar to that observed in 2D case. A reduction in positioning error or an increase in transmit power leads to a higher channel capacity across all link distance, while the influence of $\theta_{3\text{dB}}$ on capacity varies with distance.

\begin{figure}
	\centering
	\includegraphics[width=0.45\textwidth]{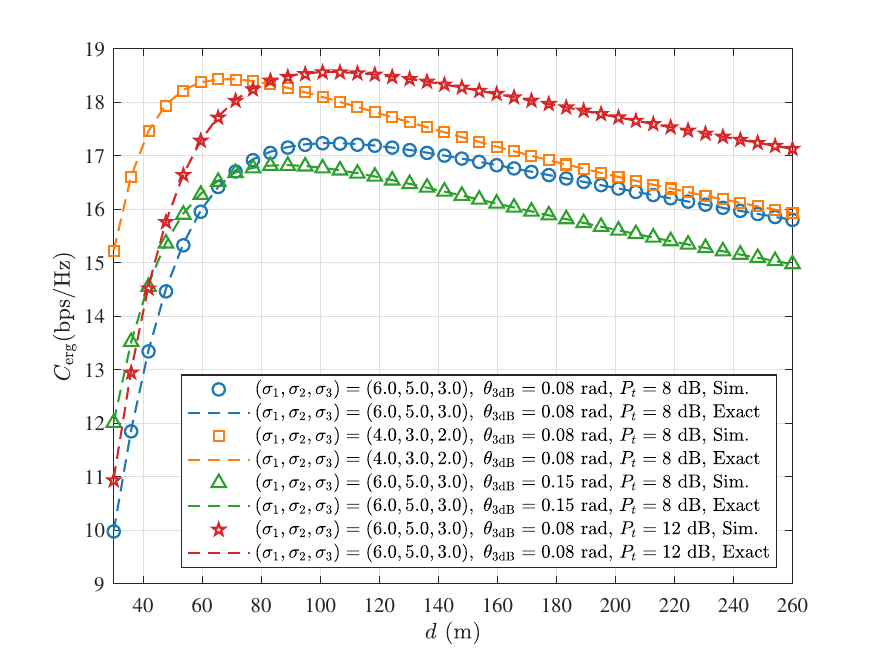}\\
	\caption{Capacity as a function of the link distance $d$ for 3D beamforming systems. $A_e = 1$~cm$^2$; $\phi_{3\text{dB}} = 0.12$~rad; $\left( \varphi_x, \varphi_y, \varphi_z \right) = \left( \pi/6, \pi/3, \pi/4 \right)$.}\label{Figure_3D_d}
\end{figure}

\begin{figure}
	\centering
	\includegraphics[width=0.45\textwidth]{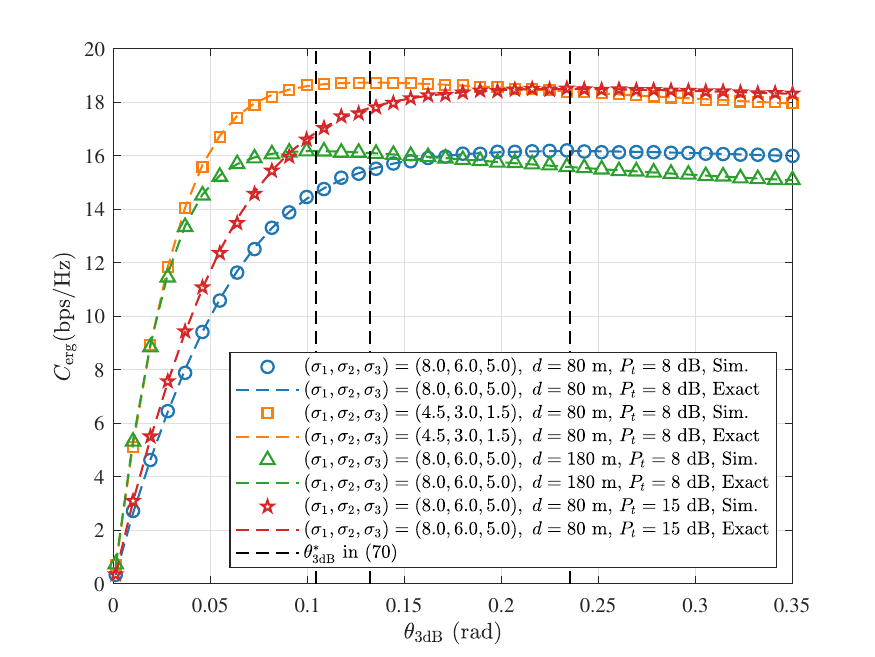}\\
	\caption{Capacity as a function of the 3-dB beamwidth $\theta_{3\text{dB}}$ for 3D beamforming systems. $A_e = 1$~cm$^2$; $\phi_{3\text{dB}} = \phi_{3\text{dB}}^\ast$; $\psi = \psi^\ast$; $\left( \varphi_x, \varphi_y, \varphi_z \right) = \left( \pi/6, \pi/3, \pi/4 \right)$. The optimal $\theta_{3\text{dB}}^\ast$ of the blue and red curves overlap.}\label{Figure_3D_theta}
\end{figure}

Figure~\ref{Figure_3D_theta} shows the capacity as the function of $\theta_{3\text{dB}}$ in 3D scenarios, where $\phi_{3\text{dB}}$ and $\psi$ are set to its optimal value according to~\eqref{Eq_3D_optimal_beamwidth}. It can be observed that different system parameters affect both the capacity and its optimal beamwidth. Specifically, reducing the positioning error increases the capacity for all values of $\theta_{3\text{dB}}$ and shifts the optimal point to a smaller beamwidth. Increasing the link distance also changes the capacity and leads to a smaller optimal point, while higher transmit power improves the overall capacity without changing the optimal beam pattern. These trends are all consistent with the analytical relationship derived in~\eqref{Eq_3D_optimal_beamwidth}.
It is also noted that the functional dependence on $\phi_{3\text{dB}}$ exhibits trends similar to those for $\theta_{3\text{dB}}$, so only the results with respect to $\theta_{3\text{dB}}$ are presented for clarity.

Figure~\ref{Figure_3D_sigma1} exhibits the channel capacity as a function of the error standard deviation $\sigma_1$ along the major axis of the PDF contours with different link distance and beam pattern for 3D beamforming systems. Similar to the 2D cases, for small divergence angles $\theta_{3\text{dB}} = 0.01$~rad, $\phi_{3\text{dB}} = 0.02$~rad, the channel capacity decreases rapidly with increasing error due to the high sensitivity of narrow beams to positioning errors. Meanwhile, for large divergence angles $\theta_{3\text{dB}} = 0.45$~rad, $\phi_{3\text{dB}} = 0.55$~rad, even large positioning errors remain within the beam coverage, making the channel capacity more robust to positioning errors.

\begin{figure}
	\centering
	\includegraphics[width=0.45\textwidth]{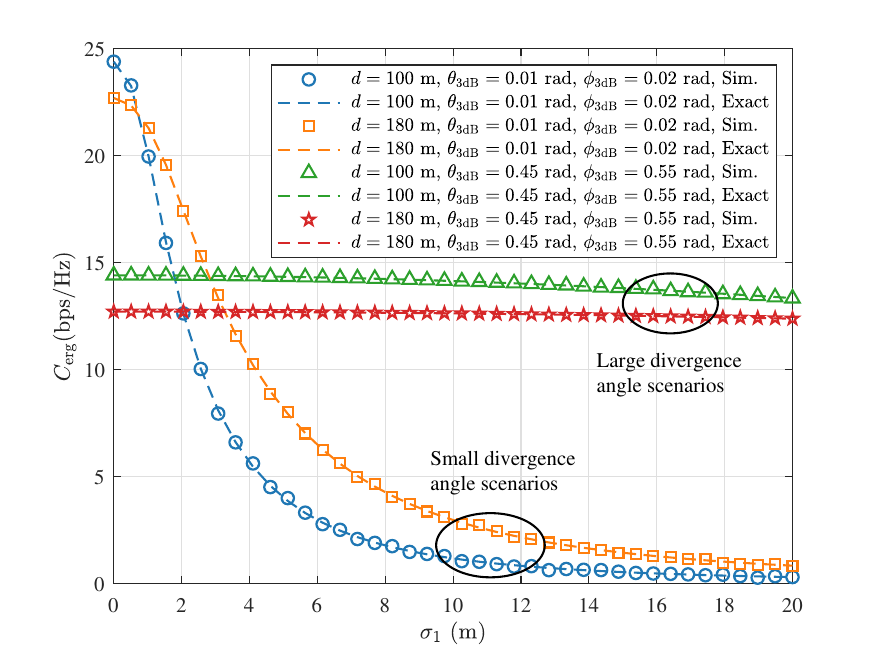}\\
	\caption{Capacity as a function of the error standard deviation $\sigma_1$ along the major axis of the PDF contours for 3D beamforming systems. $P_t = 8$~dB; $A_e = 1$~cm$^2$; $\left( \varphi_x, \varphi_y, \varphi_z \right) = \left( 0, 0, 0 \right)$; $\sigma_2 = 0.8 \, \sigma_1$; $\sigma_3 = 0.6 \, \sigma_1$.}\label{Figure_3D_sigma1}
\end{figure}

\begin{figure}
	\centering
	\includegraphics[width=0.45\textwidth]{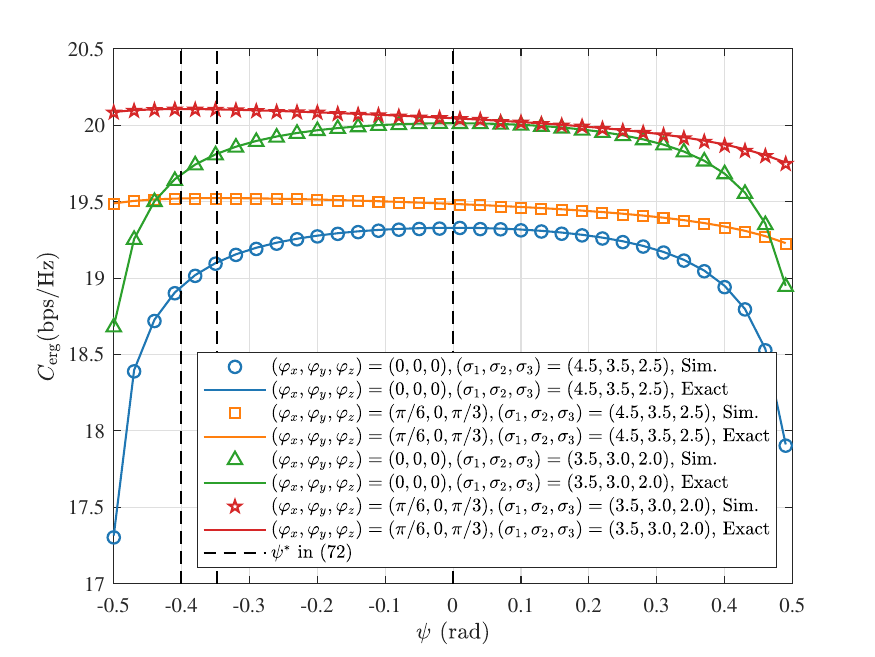}\\
	\caption{Capacity as a function of the beam rotation angle $\psi$ for 3D beamforming systems. $d = 150$~m; $P_t = 12$~dB; $A_e = 1$~cm$^2$; $\theta_{3 \text{dB}} = \theta_{3 \text{dB}}^\ast$; $\phi_{3 \text{dB}} = \phi_{3 \text{dB}}^\ast$. The optimal $\theta_{3\text{dB}}^\ast$ of the blue and green curves overlap.}\label{Figure_3D_m}
\end{figure}

Figure~\ref{Figure_3D_m} characterizes the channel capacity as a function of the beam rotation angle $\psi$ in 3D cases. As illustrated by the blue and green curves, when the directional angle of the positioning error is zero, the maximum channel capacity is achieved when the beam rotation angle is also zero. Furthermore, as demonstrated by the orange and red curves, when the directional angle of the positioning error changes, the beam rotation angle should be adjusted accordingly to maintain the maximum channel capacity. The results are consistent with the conclusion drawn from~\eqref{Eq_3D_optimal_psi}, which indicates that the beam rotation angle should be aligned with the directional angle of the contours of  positioning error ellipse in order to better match the spatial distribution of users, and thus maximize the channel capacity.

\section{Conclusion}\label{Conclusion}
We investigate the instantaneous and ergodic capacity of positioning-assisted beamforming systems, and derived their closed-form approximations by considering link distance, transmit power, positioning error and beam pattern for both 2D and 3D scenarios. Besides, we derived closed-form expressions for the optimal 2D and 3D beam pattern, which demonstrate that the channel capacity can be maximized through appropriately designed beam pattern.

{\appendices
\section{Proof of Theorem 1}\label{Appendix_Theorem1}
	\begin{IEEEproof}
		To prove Theorem~\ref{Theorem_Asymptotic_Case1}, it is equivalent to prove
		\begin{equation}\label{Eq_2D_Asymptotic_Case1_lim}
			\lim_{\frac{\theta_{3 \text{dB}} d}{\sigma_x} \to 0} \frac{ \int_{-\infty}^{+\infty} C_{\text{inst}}^{2\text{D}} \left( x \right) \cdot f_{\hat{x}_u} \left( x \right)  dx }{ f_{\hat{x}_u} (0) \cdot  \int_{-\infty}^{+\infty} C_{\text{inst}}^{2\text{D}} \left( x \right) dx } = 1.
		\end{equation}
		We use the squeeze theorem to prove~\eqref{Eq_2D_Asymptotic_Case1_lim}. Since $f_{\hat{x}_u} \left( x \right)$ is even and monotonically decreasing over $[0,+\infty)$, it satisfies $f_{\hat{x}_u} \left( x \right) \leq f_{\hat{x}_u} \left( 0 \right)$ for any $x$. Hence, we have
		\begin{equation}
			\int_{-\infty}^{+\infty} C_{\text{inst}}^{2\text{D}} \left( x \right) \cdot f_{\hat{x}_u} \left( x \right)  dx \leq f_{\hat{x}_u} \left( 0 \right) \cdot  \int_{-\infty}^{+\infty} C_{\text{inst}}^{2\text{D}} \left( x \right) dx
		\end{equation}
		and thus, the upper bound is
		\begin{equation}\label{Eq_2D_Asymptotic_Case1_UB_final}
			\frac{ \int_{-\infty}^{+\infty} C_{\text{inst}}^{2\text{D}} \left( x \right) \cdot f_{\hat{x}_u} \left( x \right)  dx }{ f_{\hat{x}_u} \left( 0 \right) \cdot \int_{-\infty}^{+\infty} C_{\text{inst}}^{2\text{D}} \left( x \right) dx } \leq 1.
		\end{equation}
		Then, we derive the lower bound. Let $\delta=\sqrt{\sigma_x \cdot \theta_{3\text{dB}} d}>0$. Since $f_{\hat{x}_u} \left( x \right)$ is even and monotonically decreasing over $[0,+\infty)$, it follows that $f_{\hat{x}_u} \left( x \right) \geq f_{\hat{x}_u} \left( \delta \right)$ for any $\left| x \right| \leq \delta$. The numerator of~\eqref{Eq_2D_Asymptotic_Case1_lim} is decomposed into two parts as
		\begin{equation}\label{Eq_2D_Asymptotic_Case1_relax}
			\begin{aligned}
				&\int_{-\infty}^{+\infty} C_{\text{inst}}^{2\text{D}} \left( x \right) \cdot f_{\hat{x}_u} \left( x \right)  dx \\
				= &\int_{\left| x \right| \leq \delta} C_{\text{inst}}^{2\text{D}} \left( x \right) \cdot f_{\hat{x}_u} \left( x \right)  dx + \int_{\left| x \right| > \delta} C_{\text{inst}}^{2\text{D}} \left( x \right) \cdot f_{\hat{x}_u} \left( x \right)  dx \\
				\geq & \int_{\left| x \right| \leq \delta} C_{\text{inst}}^{2\text{D}} \left( x \right) \cdot f_{\hat{x}_u} \left( x \right)  dx \geq  f_{\hat{x}_u} \left( \delta \right) \cdot \int_{\left| x \right| \leq \delta} C_{\text{inst}}^{2\text{D}} \left( x \right) dx.
			\end{aligned}
		\end{equation}
		Divide the left-hand and right-hand side of~\eqref{Eq_2D_Asymptotic_Case1_relax} by $f_{\hat{x}_u} \left( 0 \right) \times \int_{-\infty}^{+\infty} C_{\text{inst}}^{2\text{D}} \left( x \right) dx $, and we can obtain the lower bound as
		\begin{equation}\label{Eq_2D_Asymptotic_Case1_LB}
			\frac{ \int_{-\infty}^{+\infty} C_{\text{inst}}^{2\text{D}} \left( x \right) \cdot f_{\hat{x}_u} \left( x \right)  dx }{ f_{\hat{x}_u} \left( 0 \right) \cdot \int_{-\infty}^{+\infty} C_{\text{inst}}^{2\text{D}} \left( x \right) dx } \geq  \underbrace{\frac{f_{\hat{x}_u} \left( \delta \right)}{f_{\hat{x}_u} \left( 0 \right)}}_{T_1} \cdot \underbrace{\frac{\int_{\left| x \right| \leq \delta} C_{\text{inst}}^{2\text{D}} \left( x \right) dx}{\int_{-\infty}^{+\infty} C_{\text{inst}}^{2\text{D}} \left( x \right) dx}}_{T_2}.
		\end{equation}
		Under the asymptotic case $\theta_{3 \text{dB}} d / \sigma_x \to 0$, $T_1$ can be calculated by substituting $\delta=\sqrt{\sigma_x \cdot \theta_{3\text{dB}} d}$ into~\eqref{Eq_2D_fx} as
		\begin{equation}\label{Eq_2D_Asymptotic_Case1_T1}
			T_1 = \frac{f_{\hat{x}_u} \left( \delta \right)}{f_{\hat{x}_u} \left( 0 \right)} = \exp \left( -\frac{\delta^2}{2\sigma_x^2} \right) = \exp \left( -\frac{\theta_{3 \text{dB}} d}{2\sigma_x} \right) \to 1.
		\end{equation}
		Then we calculate $T_2$ in~\eqref{Eq_2D_Asymptotic_Case1_LB}. Let $y = x / \left( \theta_{3 \text{dB}} d \right)$ and rewrite $T_2$ as
		\begin{equation}\label{Eq_2D_Asymptotic_Case1_T2}
			\begin{aligned}
				T_2 &= \frac{\int_{\left| x \right| \leq \delta} C_{\text{inst}}^{2\text{D}} \left( x \right) dx}{\int_{-\infty}^{+\infty} C_{\text{inst}}^{2\text{D}} \left( x \right) dx} \\
				&= \frac{\int_{\left| y \right| \leq \frac{\delta}{\theta_{3 \text{dB}} d}} \log_2 \left( 1 + \frac{P_{\max} A_e}{ 4 \pi d^2 N_0}\cdot 10^{-1.2 \, y^2} \right) dy}{\int_{-\infty}^{+\infty} \log_2 \left( 1 + \frac{P_{\max} A_e}{ 4 \pi d^2 N_0}\cdot 10^{-1.2 \, y^2} \right) dy}
			\end{aligned}
		\end{equation}
		Under the asymptotic case $\theta_{3 \text{dB}} d / \sigma_x \to 0$, the integral region of the numerator in the right-hand side of~\eqref{Eq_2D_Asymptotic_Case1_T2} satisfies
		\begin{equation}\label{Eq_2D_Asymptotic_Case1_y}
			\left| y \right| \leq \frac{\delta}{\theta_{3 \text{dB}} d} = \sqrt{\frac{\sigma_x}{\theta_{3 \text{dB}} d}} \to +\infty,
		\end{equation}
		and thus, we have $T_2 \to 1$. Substituting $T_1 \to 1$ and $T_2 \to 1$ into~\eqref{Eq_2D_Asymptotic_Case1_LB}, we can calculate the lower bound as
		\begin{equation}\label{Eq_2D_Asymptotic_Case1_LB_final}
			\frac{ \int_{-\infty}^{+\infty} C_{\text{inst}}^{2\text{D}} \left( x \right) \cdot f_{\hat{x}_u} \left( x \right)  dx }{ f_{\hat{x}_u} \left( 0 \right) \cdot \int_{-\infty}^{+\infty} C_{\text{inst}}^{2\text{D}} \left( x \right) dx } \geq  1.
		\end{equation}
		Combining~\eqref{Eq_2D_Asymptotic_Case1_UB_final} and~\eqref{Eq_2D_Asymptotic_Case1_LB_final},~\eqref{Eq_2D_Asymptotic_Case1_lim} is proved.
	\end{IEEEproof}

\section{Proof of Theorem 1}\label{Appendix_Theorem2}
	\begin{IEEEproof}
		To prove Theorem~\ref{Theorem_Asymptotic_Case2}, it is equivalent to prove
		\begin{equation}\label{Eq_2D_Asymptotic_Case2_lim}
			\lim_{\frac{\theta_{3 \text{dB}} d}{\sigma_x} \to +\infty} \frac{ \int_{-\infty}^{+\infty} C_{\text{inst}}^{2\text{D}} \left( x \right) \cdot f_{\hat{x}_u} \left( x \right)  dx }{ C_{\text{inst}}^{2\text{D}} \left( 0 \right) } = 1.
		\end{equation}
		We also use the squeeze theorem to prove~\eqref{Eq_2D_Asymptotic_Case2_lim}.  Since $C_{\text{inst}}^{2\text{D}} \left( x \right)$ is even and monotonically decreasing over $[0,+\infty)$, it satisfies $C_{\text{inst}}^{2\text{D}} \left( x \right) \leq C_{\text{inst}}^{2\text{D}} \left( 0 \right)$ for any $x$. Hence, we have
		\begin{equation}
			\int_{-\infty}^{+\infty} \! C_{\text{inst}}^{2\text{D}} \! \left( x \right) \cdot f_{\hat{x}_u} \! ( x )  dx \leq  C_{\text{inst}}^{2\text{D}}  ( 0 ) \cdot  \! \int_{-\infty}^{+\infty} \! f_{\hat{x}_u} \! ( x )  dx = C_{\text{inst}}^{2\text{D}}  ( 0 ) 
		\end{equation}
		and thus, the upper bound is
		\begin{equation}\label{Eq_2D_Asymptotic_Case2_UB_final}
			\frac{ \int_{-\infty}^{+\infty} C_{\text{inst}}^{2\text{D}} \left( x \right) \cdot f_{\hat{x}_u} \left( x \right)  dx }{ C_{\text{inst}}^{2\text{D}} \left( 0 \right) } \leq 1.
		\end{equation}
		Then, we derive the lower bound. Let $\delta=\sqrt{\sigma_x \cdot \theta_{3\text{dB}} d}>0$. Since $C_{\text{inst}}^{2\text{D}} \left( x \right)$ is even and monotonically decreasing over $[0,+\infty)$, it follows that $C_{\text{inst}}^{2\text{D}} \left( x \right) \geq C_{\text{inst}}^{2\text{D}} \left( \delta \right)$ for any $\left| x \right| \leq \delta$. The numerator of~\eqref{Eq_2D_Asymptotic_Case2_lim} is decomposed into two parts as
		\begin{equation}\label{Eq_2D_Asymptotic_Case2_relax}
			\begin{aligned}
				&\int_{-\infty}^{+\infty} C_{\text{inst}}^{2\text{D}} \left( x \right) \cdot f_{\hat{x}_u} \left( x \right)  dx \\
				= &\int_{\left| x \right| \leq \delta} C_{\text{inst}}^{2\text{D}} \left( x \right) \cdot f_{\hat{x}_u} \left( x \right)  dx + \int_{\left| x \right| > \delta} C_{\text{inst}}^{2\text{D}} \left( x \right) \cdot f_{\hat{x}_u} \left( x \right)  dx \\
				\geq & \int_{\left| x \right| \leq \delta} C_{\text{inst}}^{2\text{D}} \left( x \right) \cdot f_{\hat{x}_u} \left( x \right)  dx \geq   C_{\text{inst}}^{2\text{D}}  \left( \delta \right) \cdot \int_{\left| x \right| \leq \delta} f_{\hat{x}_u} \left( x \right) dx.
			\end{aligned}
		\end{equation}
		Divide the left-hand and right-hand side of~\eqref{Eq_2D_Asymptotic_Case2_relax} by $C_{\text{inst}}^{2\text{D}} \left( 0 \right)$, and we can obtain the lower bound as
		\begin{equation}\label{Eq_2D_Asymptotic_Case2_LB}
			\frac{ \int_{-\infty}^{+\infty} C_{\text{inst}}^{2\text{D}} \left( x \right) \cdot f_{\hat{x}_u} \left( x \right)  dx }{C_{\text{inst}}^{2\text{D}} \left( 0 \right)} \geq  \underbrace{\frac{C_{\text{inst}}^{2\text{D}} \left( \delta \right)}{C_{\text{inst}}^{2\text{D}} \left( 0 \right)}}_{T_3} \cdot \underbrace{\int_{\left| x \right| \leq \delta} f_{\hat{x}_u} \left( x \right) dx}_{T_4}.
		\end{equation}
		Combining~\eqref{Eq_2D_Asymptotic_Case2_LB} with~\eqref{Eq_2D_Cinst_x},  the numerator of $T_3$ is
		\begin{equation}\label{Eq_2D_Asymptotic_Case2_T3_numerator}
			C_{\text{inst}}^{2\text{D}} \left( \delta \right) = \log_2 \left( 1 + \frac{P_{\max} A_e}{ 4 \pi d^2 N_0}\cdot 10^{-1.2 \frac{\delta^2}{d^2\theta_{3 \text{dB}}^2}} \right)
		\end{equation}
		Under the asymptotic case $\theta_{3 \text{dB}} d / \sigma_x \to +\infty$, The exponent term in~\eqref{Eq_2D_Asymptotic_Case2_T3_numerator} is
		\begin{equation}\label{Eq_2D_Asymptotic_Case2_T3_numerator_exponent}
			\frac{\delta^2}{\theta_{3 \text{dB}}^2 d^2} = \frac{\sigma_x}{\theta_{3 \text{dB}} d} \to 0
		\end{equation}
		Substituting~\eqref{Eq_2D_Asymptotic_Case2_T3_numerator_exponent} into~\eqref{Eq_2D_Asymptotic_Case2_T3_numerator}, we have $C_{\text{inst}}^{2\text{D}} \left( \delta \right) \to C_{\text{inst}}^{2\text{D}} \left( 0 \right)$, and thus $T_3 \to 1$. Under the asymptotic case $\theta_{3 \text{dB}} d / \sigma_x \to +\infty$, $T_4$ can be computed as
		\begin{equation}\label{Eq_2D_Asymptotic_Case2_T4}
			T_4 = 1 - 2 Q \left( \frac{\delta}{\sigma_x} \right) = 1 - 2 Q \left( \sqrt{\frac{\theta_{3 \text{dB}} d}{\sigma_x}} \right) \to 1.
		\end{equation}
		Substituting $T_3 \to 1$ and $T_4 \to 1$ into~\eqref{Eq_2D_Asymptotic_Case2_LB}, we can calculate the lower bound as
		\begin{equation}\label{Eq_2D_Asymptotic_Case2_LB_final}
			\frac{ \int_{-\infty}^{+\infty} C_{\text{inst}}^{2\text{D}} \left( x \right) \cdot f_{\hat{x}_u} \left( x \right)  dx }{C_{\text{inst}}^{2\text{D}} \left( 0 \right)} \geq  1.
		\end{equation}
		Combining~\eqref{Eq_2D_Asymptotic_Case2_UB_final} and~\eqref{Eq_2D_Asymptotic_Case2_LB_final}, eq.~\eqref{Eq_2D_Asymptotic_Case2_lim} is proved.
	\end{IEEEproof}

\section{Derivation of Closed-Form Expression of  $M_0^{3\text{D}}$}\label{Appendix_M0^3D}
$M_0^{3\text{D}}$ in~\eqref{Eq_3D_Cerg_infty} can be written as
\begin{equation}\label{Eq_3D_M0_initial}
	M_0^{3\text{D}} = \iint_{\mathcal{D}_0} f_{\hat{x}_u, \hat{z}_u} \left( x, z \right) dxdz
\end{equation}
where $f_{\hat{x}_u, \hat{z}_u} \left( x, z \right)$, $\mathcal{D}_0$ are defined in~\eqref{Eq_3D_fxz},~\eqref{Eq_3D_R0^2}, respectively. First, we perform the eigenvalue decomposition $\mathbf{A} = \mathbf{U} \mathbf{Y} \mathbf{U}^T$ and introduce $\mathbf{t} = \sqrt{\mathbf{Y}} \mathbf{U}^T \left[ x, z \right]^T$, under which $M_0^{3\text{D}}$ can be rewritten as
\begin{equation}\label{Eq_3D_M0_t}
	\begin{aligned}
		M_0^{3\text{D}} &= \iint\limits_{\mathbf{t}^T \mathbf{t} \leq R_0^2}
		\frac{1}{2 \pi \left| \mathbf{Y} \mathbf{\Sigma}^{'2\text{D}} \right|^{\frac{1}{2}}} \\
		&\quad \times \exp \left( -\frac{\mathbf{t}^T \sqrt{\mathbf{Y}^{-1}} \mathbf{U}^T \left( \mathbf{\Sigma}^{'2\text{D}} \right)^{-1} \mathbf{U} \sqrt{\mathbf{Y}^{-1}} \mathbf{t}}{2} \right) d\mathbf{t}.
	\end{aligned}
\end{equation}
Define $\mathbf{H}^{-1} = \sqrt{\mathbf{Y}^{-1}} \mathbf{U}^T \left( \mathbf{\Sigma}^{'2\text{D}} \right)^{-1} \mathbf{U} \sqrt{\mathbf{Y}^{-1}}$, which admits the eigenvalue decomposition $\mathbf{H} = \mathbf{V} \mathbf{D} \mathbf{V}^T$ with $\mathbf{V}$ being an orthogonal matrix and $\mathbf{D} = \left[ \begin{smallmatrix} \lambda_1^2 & 0 \\ 0 & \lambda_2^2 \end{smallmatrix} \right]$ a diagonal matrix of eigenvalues satisfying $\lambda_1 \leq \lambda_2$. Let $\mathbf{q} = \mathbf{V}^T \mathbf{t}$ and substitute it into~\eqref{Eq_3D_M0_t}, we have
\begin{equation}\label{Eq_3D_M0_q}
	M_0^{3\text{D}} = \iint\limits_{\mathbf{q}^T \mathbf{q} \leq R_0^2}
	\frac{1}{2 \pi \left| \mathbf{D} \right|^{\frac{1}{2}}}\exp \left( -\frac{\mathbf{q}^T \mathbf{D}^{-1} \mathbf{q}}{2} \right) d\mathbf{q}.
\end{equation}
Equation~\eqref{Eq_3D_M0_q} can be divided into two cases according to whether $\lambda_1$ and $\lambda_2$ are identical.

\subsubsection{$\lambda_1 < \lambda_2$}
Define $q= \lambda_1 / \lambda_2 \in \left( 0,1 \right)$ and $\omega = \lambda_1^2 + \lambda_2^2$~\cite{jajszczyk2001digital}, so that~\eqref{Eq_3D_M0_q} can be expressed as~\cite{zhu2018new}
\begin{equation}\label{Eq_3D_M0_hoyt}
	M_0^{3\text{D}} = \int_0^{R_0} h\left( r; q, \omega \right) dr
\end{equation}
where $h \left( r; q, \omega \right)$ is the PDF of Hoyt distribution with $r^2 = \mathbf{q}^T \mathbf{q}$. Furthermore, the cumulative distribution function~(CDF) of Hoyt distribution can be computed as~\cite{lopez2015asymptotically}
\begin{equation}\label{Eq_CDF_Hoyt}
	\begin{aligned}
		H \left( r; q, \omega \right) &= \int_0^r h\left( r; q, \omega \right)  dr \\
		&= Q_1 \left( \frac{r \sqrt{1 - q^4}}{2 q \sqrt{\omega}} \sqrt{\frac{1 + q}{1 - q}}, \frac{r \sqrt{1 - q^4}}{2 q \sqrt{\omega}} \sqrt{\frac{1 - q}{1 + q}} \right) \\
		&- Q_1 \left( \frac{r \sqrt{1 - q^4}}{2 q \sqrt{\omega}} \sqrt{\frac{1 - q}{1 + q}}, \frac{r \sqrt{1 - q^4}}{2 q \sqrt{\omega}} \sqrt{\frac{1 + q}{1 - q}} \right)
	\end{aligned}
\end{equation}
where $Q_1 \left( a, b \right) = \int_b^{+\infty} x \exp \left( - \frac{x^2 + a^2}{2} \right) I_0 \left( a x \right) \, dx$ is the first order Marcum $Q$-function; $I_0 \left( x \right)$ is the first order modified Bessel function of first kind. Substituting~\eqref{Eq_CDF_Hoyt} into~\eqref{Eq_3D_M0_hoyt}, we have
\begin{equation}\label{Eq_3D_M0_lambda_unequal}
	\begin{aligned}
		M_0^{3\text{D}} &\! = \! Q_1 \! \left( \! \frac{R_0 \sqrt{\left( 1 - q^4 \right)}}{2 q \sqrt{\omega}} \sqrt{\frac{1 + q}{1 - q}},  \frac{R_0 \sqrt{\left( 1 - q^4 \right)}}{2 q \sqrt{\omega}} 	\sqrt{\frac{1 - q}{1 + q}} \right) \\
		&\! - \! Q_1 \! \left( \frac{R_0 \sqrt{\left( \! 1 - q^4 \right)}}{2 q \sqrt{\omega}} \sqrt{\frac{1 - q}{1 + q}},  \frac{R_0 \sqrt{\left( 1 - q^4 \right)}}{2 q \sqrt{\omega}} 	\sqrt{\frac{1 + q}{1 - q}} \right) \! .
	\end{aligned}
\end{equation}

\subsubsection{$\lambda_1 = \lambda_2$}
Define $\lambda = \lambda_1 = \lambda_2$, we can calculate~\eqref{Eq_3D_M0_q} as
\begin{equation}\label{Eq_3D_M0_lambda_equal}
	M_0^{3\text{D}} = \iint\limits_{\mathbf{q}^T \mathbf{q} \leq R_0^2} \frac{1}{2 \pi \lambda^2}\exp \left( -\frac{\mathbf{q}^T \mathbf{q}}{2\lambda^2} \right) d\mathbf{q} = 1 - \exp \left( - \frac{R_0^2}{2 \lambda^2} \right) \! .
\end{equation}


\section{Derivation of Closed-Form Expression of  $M_2^{3\text{D}}$}\label{Appendix_M2^3D}
According to~\eqref{Eq_3D_Cerg_infty}, $M_2^{3\text{D}}$ can be computed as
\begin{equation}\label{Eq_3D_M2_q}
	M_2^{3\text{D}} = \iint\limits_{\mathbf{q}^T \mathbf{q} \leq R_0^2}
	\mathbf{q}^T \mathbf{q} \cdot \frac{1}{2 \pi \left| \mathbf{D} \right|^{\frac{1}{2}}}\exp \left( -\frac{\mathbf{q}^T \mathbf{D}^{-1} \mathbf{q}}{2} \right) d\mathbf{q}
\end{equation}
and can also be divided into two cases based on if $\lambda_1 = \lambda_2$.

\setcounter{subsubsection}{0}
\subsubsection{$\lambda_1 < \lambda_2$}
Equation~\eqref{Eq_3D_M2_q} can be evaluated as
\begin{equation}\label{Eq_3D_M2_hoyt}
	M_2^{3\text{D}} = \int_0^{R_0^2} s f_s \left( s \right) ds 
\end{equation}
where $s  = \mathbf{q}^T \mathbf{q}$ and $f_s \left( s \right)$ denotes the PDF of $s$ as~\cite{lopez2015asymptotically}
\begin{equation}\label{Eq_3D_s_PDF}
	f_s \left( s \right) = \frac{1 + q^2}{2 q \omega} \exp \left( -\frac{\left( 1 + q^2 \right)^2}{4 q^2 \omega} s \right) I_0\left( \frac{1 - q^4}{4 q^2 \omega} s \right).
\end{equation}
Since $M_0^{3\text{D}}$ is available in closed form, the derivation of $M_2^{3\text{D}}$ can be simplified by first establishing a relationship between $M_0^{3\text{D}}$ and $M_2^{3\text{D}}$. We take the derivative of $f_s \left( s \right)$ based on the differentiation rule of the modified Bessel function as~\cite{abramowitz1965handbook}
\begin{equation}\label{Eq_3D_differentiation1}
	\begin{aligned}
		&\frac{df_s \left( s \right)}{ds} = -\frac{\left( 1 + q^2 \right)^2}{4 q^2 \omega} f_s \left( s \right) + \frac{\left( 1 + q^2 \right)^2 \left( 1 - q^2 \right)}{8 q^3 \omega^2} \\
		&\qquad \qquad \times \exp \left( -\frac{\left( 1 + q^2 \right)^2}{4 q^2 \omega} s \right) I_1 \left( \frac{1 - q^4}{4 q^2 \omega} s \right) .
	\end{aligned}
\end{equation}
Multiplying both sides of~\eqref{Eq_3D_differentiation1} by $s$ and then integrating over $\left[ 0, R_0^2\right]$, we apply integration by parts to the left-hand side of~\eqref{Eq_3D_differentiation1}, giving
\begin{equation}\label{Eq_3D_differentiation1_left}
	\begin{aligned}
		\int_0^{R_0^2} s \frac{df_s \left( s \right)}{ds} ds &= s f_s \left( s \right) \Big|_0^{R_0^2} - \int_0^{R_0^2} f_s \left( s \right) ds \\
		&= R_0^2 \, f_s \left( R_0^2 \right) - M_0^{3\text{D}}.
	\end{aligned}
\end{equation}
The right-hand side of~\eqref{Eq_3D_differentiation1} becomes
\begin{equation}\label{Eq_3D_differentiation1_right}
	\begin{aligned}
		&-\frac{\left( 1 + q^2 \right)^2}{4 q^2 \omega} \int_0^{R_0^2} s f_s \left( s \right) ds + \frac{\left( 1 + q^2 \right)^2 \left( 1 - q^2 \right)}{8 q^3 \omega^2} \\
		&\times \underbrace{\int_{0}^{R_0^2} \! s \exp \! \left( -\frac{\left( 1 + q^2 \right)^2}{4 q^2 \omega} s \right)  I_1 \left( \frac{1 - q^4}{4 q^2 \omega} s \right) \! ds}_{\Phi \left( R_0,q,\omega \right)} \\
		&= -\frac{\left( 1 + q^2 \right)^2}{4 q^2 \omega} M_2^{3\text{D}} + \frac{\left( 1 + q^2 \right)^2 \left( 1 - q^2 \right)}{8 q^3 \omega^2} \Phi \left( R_0,q,\omega \right)
	\end{aligned}
\end{equation}
Equating~\eqref{Eq_3D_differentiation1_left} and~\eqref{Eq_3D_differentiation1_right} and combining with~\eqref{Eq_3D_s_PDF}, we obtain
\begin{equation}\label{Eq_3D_differentiation1_s_result}
	\begin{aligned}
		& \frac{\left( 1 + q^2 \right)^2}{4 q^2 \omega} M_2^{3\text{D}} - M_0^{3\text{D}} = \frac{\left( 1 + q^2 \right)^2 \left( 1 - q^2 \right)}{8 q^3 \omega^2} \Phi \left( R_0,q,\omega \right) \\
		&- \frac{\left( 1 + q^2 \right) R_0^2 }{2 q \omega}  \exp \left( - \frac{\left( 1 + q^2 \right)^2 R_0^2}{4 q^2 \omega } \right) I_0 \left(  \frac{ \left( 1 - q^4 \right) R_0^2 }{4 q^2 \omega} \right).
	\end{aligned}
\end{equation}
Equation~\eqref{Eq_3D_differentiation1_s_result} establishes the relationship between $M_0^{3\text{D}}$ and $M_2^{3\text{D}}$, but it still involves integrals, i.e., $\Phi \left( R_0,q,\omega \right)$. To evaluate it, we take the derivative of its integrand as~\cite{abramowitz1965handbook}
\begin{equation}\label{Eq_3D_differentiation2}
	\begin{aligned}
		&\frac{d}{ds} \left[ s \exp \left( -\frac{\left( 1 + q^2 \right)^2}{4 q^2 \omega} s \right) I_1 \left( \frac{1 - q^4}{4 q^2 \omega} s \right) \right] \\
		=& \frac{1 - q^4}{4 q^2 \omega} s \exp \left( -\frac{\left( 1 + q^2 \right)^2}{4 q^2 \omega} s \right) I_0 \left( \frac{1 - q^4}{4 q^2 \omega} s \right) \\
		-&\frac{\left( 1 + q^2 \right)^2}{4 q^2 \omega} s \exp \left( -\frac{\left( 1 + q^2 \right)^2}{4 q^2 \omega} s \right) I_1 \left( \frac{1 - q^4}{4 q^2 \omega} s \right).
	\end{aligned}
\end{equation}
Integrating both sides of~\eqref{Eq_3D_differentiation2} over $\left[ 0, R_0^2\right]$ and applying the Newton-Leibniz theorem to the left-hand side of~\eqref{Eq_3D_differentiation2} gives
\begin{equation}\label{Eq_3D_differentiation2_left}
	\begin{aligned}
		&\int_0^{R_0^2} \frac{d}{ds} \left[ s \exp \left( -\frac{\left( 1 + q^2 \right)^2}{4 q^2 \omega} s \right) I_1 \left( \frac{1 - q^4}{4 q^2 \omega} s \right) \right] ds \\
		= &\left. s \exp \left( -\frac{\left( 1 + q^2 \right)^2}{4 q^2 \omega} s \right) I_1 \left( \frac{1 - q^4}{4 q^2 \omega} s \right)  \right|_0^{R_0^2} \\
		= & \, R_0^2 \exp \left( -\frac{\left( 1 + q^2 \right)^2 R_0^2}{4 q^2 \omega} \right) I_1 \left( \frac{\left( 1 - q^4 \right) R_0^2}{4 q^2 \omega} \right),
	\end{aligned}
\end{equation}
while the right-hand side of~\eqref{Eq_3D_differentiation2} becomes
\begin{equation}\label{Eq_3D_differentiation2_right}
	\begin{aligned}
		& \frac{1 - q^4}{4 q^2 \omega} \int_0^{R_0^2} s \exp \left( -\frac{\left( 1 + q^2 \right)^2}{4 q^2 \omega} s \right) I_0 \left( \frac{1 - q^4}{4 q^2 \omega} s \right) ds \\
		-&\frac{\left( 1 + q^2 \right)^2}{4 q^2 \omega} \int_0^{R_0^2}  s \exp \left( -\frac{\left( 1 + q^2 \right)^2}{4 q^2 \omega} s \right) I_1 \left( \frac{1 - q^4}{4 q^2 \omega} s \right) ds \\
		=& \frac{1 - q^2}{2 q} M_2^{3\text{D}} - \frac{\left( 1 + q^2 \right)^2}{4 q^2 \omega} \Phi \left( R_0,q,\omega \right).
	\end{aligned}
\end{equation}
Equating \eqref{Eq_3D_differentiation2_left} and \eqref{Eq_3D_differentiation2_right} 
to solve $\Phi \left( R_0,q,\omega \right)$, we obtain
\begin{equation}\label{Eq_3D_differentiation2_result}
	\begin{aligned}
		&\Phi \left( R_0,q,\omega \right) = \frac{2 q \omega \left(1 - q^2 \right)}{\left( 1 + q^2 \right)^2} M_2^{3\text{D}} \\
		&- \frac{4 q^2 \omega R_0^2}{\left( 1 + q^2 \right)^2} \exp  \left( -\frac{\left( 1 + q^2 \right)^2 R_0^2}{4 q^2 \omega}  \right) I_1 \left( \frac{ \left( 1 - q^4 \right) R_0^2 }{4 q^2 \omega}\right) \! .
	\end{aligned}
\end{equation}
Substituting~\eqref{Eq_3D_differentiation2_result} into~\eqref{Eq_3D_differentiation1_s_result} to eliminate $\Phi \left( R_0,q,\omega \right)$, we can calculate $M_2^{3\text{D}}$ in closed form as
\begin{equation}\label{Eq_3D_M2_lambda_unequal}
	\begin{aligned}
		&M_2^{3\text{D}} =   M_0^{3\text{D}} \omega - \frac{2 q \omega R_0^2}{1 + q^2} \exp \left( -\frac{\left( 1 + q^2 \right)^2 R_0^2}{4 q^2 \omega} \right) \\
		&\times \! \left[ \frac{\left( 1 + q^2 \right)^2}{4 q^2 \omega} I_0 \! \left( \frac{\left( 1 - q^4\right) \! R_0^2}{4 q^2 \omega} \right) \! + \! \frac{1 - q^4}{4 q^2 \omega} I_1 \! \left( \! \frac{\left( 1 - q^4\right) \! R_0^2}{4 q^2 \omega}\! \right) \! \right]\!.
	\end{aligned}
\end{equation}

\subsubsection{$\lambda_1 = \lambda_2$}
Define $\lambda = \lambda_1 = \lambda_2$, we can calculate~\eqref{Eq_3D_M2_q} as
\begin{equation}\label{Eq_3D_M2_lambda_equal}
	\begin{aligned}
		M_2^{3\text{D}} &= \iint\limits_{\mathbf{q}^T \mathbf{q} \leq R_0^2} \mathbf{q}^T \mathbf{q} \cdot \frac{1}{2 \pi \lambda^2}\exp \left( -\frac{\mathbf{q}^T \mathbf{q}}{2\lambda^2} \right) d\mathbf{q} \\
		&= \frac{1}{\lambda^2} \int_0^{R_0} r^3 \exp \left( -\frac{r^2}{2 \lambda^2} \right) dr \\
		&= 2 \lambda^2 \left[ 1 - \left( 1 + \frac{R_0^2}{2 \lambda^2} \right) \exp \left( -\frac{R_0^2}{2 \lambda^2} \right) \right].
	\end{aligned}
\end{equation}


\section{Derivation of Closed-Form Expression of  $M_4^{3\text{D}}$}\label{Appendix_M4^3D}

With $M_0^{3\text{D}}$ and $M_2^{3\text{D}}$ in closed-form derived, $M_4^{3\text{D}}$ can be obtained as follows. First, $M_4^{3\text{D}}$ can be rewritten as
\begin{equation}\label{Eq_3D_M4_q}
	M_4^{3\text{D}} = \iint\limits_{\mathbf{q}^T \mathbf{q} \leq R_0^2}
	\left( \mathbf{q}^T \mathbf{q} \right)^2 \cdot \frac{1}{2 \pi \left| \mathbf{D} \right|^{\frac{1}{2}}}\exp \left( -\frac{\mathbf{q}^T \mathbf{D}^{-1} \mathbf{q}}{2} \right) d\mathbf{q}.
\end{equation}
Equation~\eqref{Eq_3D_M4_q} can also be categorized into two cases based on whether $\lambda_1 = \lambda_2$.

\setcounter{subsubsection}{0}
\subsubsection{$\lambda_1 < \lambda_2$}
Similar to~\eqref{Eq_3D_M2_hoyt}, we rewrite~\eqref{Eq_3D_M4_q} as
\begin{equation}\label{Eq_3D_M4_s}
	M_4^{3\text{D}} = \int_0^{R_0^2} s^2 f_s \left( s \right) ds.
\end{equation}
To evaluate~\eqref{Eq_3D_M4_s}, it is first necessary to construct the expression of its integrand, i.e., $s^2 f_s \left( s \right)$. We multiply both sides of~\eqref{Eq_3D_differentiation1} by $s^2$ and integrate over $\left[ 0, R_0^2\right]$, and the left-hand side of~\eqref{Eq_3D_differentiation1} is evaluated by integration by parts, giving
\begin{equation}\label{Eq_3D_differentiation1_s2_left}
	\begin{aligned}
		\int_0^{R_0^2} s^2 \frac{df_s \left( s \right)}{ds} ds &= s^2 f_s \left( s \right) \Big|_0^{R_0^2} - 2 \int_0^{R_0^2} s f_s \left( s \right) ds \\
		&= R_0^4 \, f_s \left( R_0^2 \right) - 2 M_2^{3\text{D}},
	\end{aligned}
\end{equation}
while the right-hand side of~\eqref{Eq_3D_differentiation1} can be written as
\begin{equation}\label{Eq_3D_differentiation1_s2_right}
	\begin{aligned}
		&-\frac{\left( 1 + q^2 \right)^2}{4 q^2 \omega} \int_0^{R_0^2} s^2 f_s \left( s \right) ds + \frac{\left( 1 + q^2 \right)^2 \left( 1 - q^2 \right)}{8 q^3 \omega^2} \\
		&\times \underbrace{\int_{0}^{R_0^2}  s^2 \exp \left( -\frac{\left( 1 + q^2 \right)^2}{4 q^2 \omega} s \right)  I_1 \left( \frac{1 - q^4}{4 q^2 \omega} s \right) \! ds}_{\Psi \left( R_0,q,\omega \right)} \\
		&= -\frac{\left( 1 + q^2 \right)^2}{4 q^2 \omega} M_4^{3\text{D}} + \frac{\left( 1 + q^2 \right)^2 \left( 1 - q^2 \right)}{8 q^3 \omega^2} \Psi \left( R_0,q,\omega \right).
	\end{aligned}
\end{equation}
Equating~\eqref{Eq_3D_differentiation1_s2_left} and~\eqref{Eq_3D_differentiation1_s2_right} and substituting~\eqref{Eq_3D_s_PDF} gives
\begin{equation}\label{Eq_3D_differentiation1_s2_result}
	\begin{aligned}
		& \frac{\left( 1 + q^2 \right)^2}{4 q^2 \omega} M_4^{3\text{D}} - 2 M_2^{3\text{D}} = \frac{\left( 1 + q^2 \right)^2 \left( 1 - q^2 \right)}{8 q^3 \omega^2} \Psi \left( R_0,q,\omega \right) \\
		&- \frac{\left( 1 + q^2 \right) R_0^4 }{2 q \omega}  \exp \left( - \frac{\left( 1 + q^2 \right)^2 R_0^2}{4 q^2 \omega } \right) I_0 \left(  \frac{ \left( 1 - q^4 \right) R_0^2 }{4 q^2 \omega} \right).
	\end{aligned}
\end{equation}
Equation~\eqref{Eq_3D_differentiation1_s2_result} establishes the relationship between $M_2^{3\text{D}}$ and $M_4^{3\text{D}}$, but $\Psi \left( R_0,q,\omega \right)$ still involves integrals to be calculated. Multiplying both sides of~\eqref{Eq_3D_differentiation2} by $s$ and integrating over $\left[ 0, R_0^2\right]$, we can compute the left-hand side of~\eqref{Eq_3D_differentiation2} with integration by parts as
\begin{equation}\label{Eq_3D_differentiation2_s2_left}
	\begin{aligned}
		&\int_0^{R_0^2} s \frac{d}{ds} \left[ s \exp \left( -\frac{\left( 1 + q^2 \right)^2}{4 q^2 \omega} s \right) I_1 \left( \frac{1 - q^4}{4 q^2 \omega} s \right) \right] ds \\
		= &\left. s^2 \exp \left( -\frac{\left( 1 + q^2 \right)^2}{4 q^2 \omega} s \right) I_1 \left( \frac{1 - q^4}{4 q^2 \omega} s \right)  \right|_0^{R_0^2} \\
		-& \int_0^{R_0^2} s \exp \left( -\frac{\left( 1 + q^2 \right)^2}{4 q^2 \omega} s \right) I_1 \left( \frac{1 - q^4}{4 q^2 \omega} s \right) ds \\
		=& R_0^4 \exp \left( \! -\frac{\left( 1 + q^2 \right)^2 \! R_0^2}{4 q^2 \omega} \! \right) I_1 \! \left( \! \frac{ \left( 1 - q^4 \right) R_0^2}{4 q^2 \omega} \! \right) \! - \! \Phi \left( R_0,q,\omega \right) \! .
	\end{aligned}
\end{equation}
The right-hand side of~\eqref{Eq_3D_differentiation2} simplifies to
\begin{equation}\label{Eq_3D_differentiation2_s2_right}
	\begin{aligned}
		& \frac{1 - q^4}{4 q^2 \omega} \int_0^{R_0^2} s^2 \exp \left( -\frac{\left( 1 + q^2 \right)^2}{4 q^2 \omega} s \right) I_0 \left( \frac{1 - q^4}{4 q^2 \omega} s \right) ds \\
		-&\frac{\left( 1 + q^2 \right)^2}{4 q^2 \omega} \int_0^{R_0^2}  s^2 \exp \left( -\frac{\left( 1 + q^2 \right)^2}{4 q^2 \omega} s \right) I_1 \left( \frac{1 - q^4}{4 q^2 \omega} s \right) ds \\
		=& \frac{1 - q^2}{2 q} M_4^{3\text{D}} - \frac{\left( 1 + q^2 \right)^2}{4 q^2 \omega} \Psi \left( R_0,q,\omega \right).
	\end{aligned}
\end{equation}
Setting~\eqref{Eq_3D_differentiation2_s2_left} equal to~\eqref{Eq_3D_differentiation2_s2_right} and solving for $\Psi \left( R_0,q,\omega \right)$, we have
\begin{equation}\label{Eq_3D_differentiation2_s2_result}
	\begin{aligned}
		&\Psi \left( R_0,q,\omega \right) = \frac{4 q^2 \omega}{\left( 1 + q^2 \right)^2} \Phi \left( R_0,q,\omega \right) + \frac{2 q \omega \left( 1 - q^2 \right)}{ \left( 1 + q^2 \right)^2} M_4^{3\text{D}}  \\
		&- \frac{4 q^2 \omega R_0^4}{\left( 1 + q^2 \right)^2} \exp \! \left( \! -\frac{\left( 1 + q^2 \right)^2 \! R_0^2}{4 q^2 \omega} \right) \! I_1 \left(  \frac{ \left( 1 - q^4 \right) R_0^2 }{4 q^2 \omega} \right) \! .
	\end{aligned}
\end{equation}
Substituting~\eqref{Eq_3D_differentiation2_s2_result} into~\eqref{Eq_3D_differentiation1_s2_result} to eliminate $\Psi \left( R_0,q,\omega \right)$, we can calculate $M_4^{3\text{D}}$ in closed form as
\begin{equation}\label{Eq_3D_M4_lambda_unequal}
	\begin{aligned}
		M_4^{3\text{D}} \! &= \! \frac{\left( 3 q^4 \! + \! 2 q^2 \! + \! 3 \right) \! \omega}{ \left( 1 + q^2 \right)^2} M_2^{3\text{D}} \! - \! \frac{2 q \omega R_0^2}{1 + q^2} \exp \! \left( \! -\frac{\left( 1 + q^2 \right)^{\! 2} \! R_0^2}{4 q^2 \omega} \! \right) \\ 
		&\times \left[  \frac{\left( 1 + q^2 \right)^2 R_0^2}{4 q^2 \omega} I_0 \left( \frac{\left( 1 - q^4\right) R_0^2}{4 q^2 \omega} \right) \right. \\
		&\quad \left. + \frac{1 - q^4}{4 q^2 \omega} \left( R_0^2 + \frac{4 q^2 \omega}{\left( 1 + q^2 \right)^2} \right) I_1  \left( \frac{\left( 1 - q^4\right) R_0^2}{4 q^2 \omega} \right) \right].
	\end{aligned}
\end{equation}

\subsubsection{$\lambda_1 = \lambda_2$}
Define $\lambda = \lambda_1 = \lambda_2$, we can calculate~\eqref{Eq_3D_M4_q} as
\begin{equation}\label{Eq_3D_M4_lambda_equal}
	\begin{aligned}
		M_4^{3\text{D}} &= \ \iint\limits_{\mathbf{q}^T \mathbf{q} \leq R_0^2}
		\left( \mathbf{q}^T \mathbf{q} \right)^2 \cdot \frac{1}{2 \pi \left| \mathbf{D} \right|^{\frac{1}{2}}}\exp \left( -\frac{\mathbf{q}^T \mathbf{D}^{-1} \mathbf{q}}{2} \right) d\mathbf{q} \\
		&= \frac{1}{2 \lambda^2} \int_0^{R_0} r^5 \exp \left( -\frac{r^2}{2 \lambda^2} \right) dr \\
		&= 8 \lambda^4 \left[ 1 - \left( 1 + \frac{R_0^2}{2 \lambda^2} + \frac{R_0^4}{8 \lambda^4} \right) \exp \left( -\frac{R_0^2}{2 \lambda^2} \right) \right].
	\end{aligned}
\end{equation}


\bibliographystyle{IEEEtran}
\bibliography{REF_paper2}

@article{thallapalli2026gridless,
	title={Gridless Sparse Channel Estimation for Dual Wideband THz Ultra-Massive {MIMO} Systems},
	author={Thallapalli, Soujanya and Sen, Debarati},
	journal={IEEE Trans. Wireless Commun.},
	volume={25},
	pages={10035--10050},
	year={2026},
	publisher={IEEE}
}

@article{zhang2026fundamental,
	title={On fundamental limits for fluid antenna-assisted integrated sensing and communications for unsourced random access},
	author={Zhang, Zhentian and Wong, Kai-Kit and Dang, Jian and Zhang, Zaichen and Chae, Chan-Byoung},
	journal={IEEE J. Sel. Areas Commun.},
	year={2026},
	publisher={IEEE}
}

@article{tomasin2025two,
  title={Two-Way Protocol for Downlink Channel State Information Reporting in {FDD} Networks},
  author={Tomasin, Stefano and Merlo, Simone},
  journal={IEEE Trans. Veh. Technol.},
  year={2025},
  volume={74},
  number={4},
  pages={6829--6833},
  publisher={IEEE}
}

@article{becirovic2022combining,
  title={Combining reciprocity and {CSI} feedback in {MIMO} systems},
  author={Becirovic, Ema and Bj{\"o}rnson, Emil and Larsson, Erik G},
  journal={IEEE Trans. Wireless Commun.},
  volume={21},
  number={11},
  pages={10065--10080},
  year={2022},
  publisher={IEEE}
}

@article{talvitie2019positioning,
	title={Positioning and location-aware communications for modern railways with {5G} new radio},
	author={Talvitie, Jukka and Levanen, Toni and Koivisto, Mike and Ihalainen, Tero and Pajukoski, Kari and Valkama, Mikko},
	journal={IEEE Commun. Mag.},
	volume={57},
	number={9},
	pages={24--30},
	year={2019},
	publisher={IEEE}
}

@article{zhang2026positioning,
	title={Positioning-Aided Channel Estimation for Multi-{LEO} Satellite Cooperative Beamforming},
	author={Zhang, Yuchen and Zheng, Pinjun and Ma, Jie and Wymeersch, Henk and Al-Naffouri, Tareq Y},
	journal={IEEE Trans. Commun.}, 
	year={2026},
	volume={74},
	pages={3888-3903},
	publisher={IEEE}
}

@inproceedings{maiberger2010location,
	title={Location based beamforming},
	author={Maiberger, Roy and Ezri, Doron and Erlihson, Michael},
	booktitle={Proc. IEEE 26th Conv. Electr. Electron. Eng. Israel},
	pages={184--187},
	year={2010},
}

@article{muppirisetty2018location,
	title={Location-aided pilot contamination avoidance for massive {MIMO} systems},
	author={Muppirisetty, L Srikar and Charalambous, Themistoklis and Karout, Johnny and Fodor, G{\'a}bor and Wymeersch, Henk},
	journal={IEEE Trans. Wireless Commun.},
	volume={17},
	number={4},
	pages={2662--2674},
	year={2018},
	publisher={IEEE}
}

@article{talvitie2020beamformed,
	title={Beamformed radio link capacity under positioning uncertainty},
	author={Talvitie, Jukka and Levanen, Toni and Koivisto, Mike and Ihalainen, Tero and Pajukoski, Kari and Valkama, Mikko},
	journal={IEEE Trans. Veh. Technol.},
	volume={69},
	number={12},
	pages={16235--16240},
	year={2020},
	publisher={IEEE}
}

@article{shi2023joint,
	title={Joint beamformer design and power allocation method for hybrid {RF-VLCP} system},
	author={Shi, Shengnan and Gui, Guan and Lin, Yun and Yuen, Chau and Dobre, Octavia A and Adachi, Fumiyuki},
	journal={IEEE Internet Things J.},
	volume={11},
	number={5},
	pages={7878--7892},
	year={2023},
	publisher={IEEE}
}

@article{song2024position,
	title={Position-based adaptive beamforming and roadside unit sectorization for {V2I} communications},
	author={Song, Jiho and Lee, Jong-Ho and Noh, Song},
	journal={IEEE Trans. Veh. Technol.},
	volume={73},
	number={2},
	pages={2960--2965},
	year={2024},
	publisher={IEEE}
}

@article{he2026bistatic,
	title={Bistatic-Enhancement {MIMO} {ISAC}: Joint Beamforming Design in Cell-Free Communication and Bistatic Radar Systems},
	author={He, Boxin and Mao, Wencan and Liu, Yaxi and Huangfu, Wei and Wang, Fangxin and Zhang, Haijun},
	journal={IEEE Trans. Wireless Commun.},
	year={2026},
	volume={25},
	pages={5935--5951},
	publisher={IEEE}
}

@article{jing2026beamforming,
	title={Beamforming Designs for Multiple {UAV} Interference Systems With {LOS} Channels},
	author={Jing, Yindi and Yu, Xinwei},
	journal={IEEE Trans. Commun.},
	volume={74},
	pages={1336--1349},
	year={2026},
	publisher={IEEE}
}

@article{cao2026location,
	title={Location-Sensing-Based Beamforming for Multi-{IRS}-Aided {OFDM} Communication Systems},
	author={Cao, Xueyan and Cui, Jun and Wang, Shubin},
	journal={IEEE Trans. Green Commun. Networking},
	year={2026},
	volume={10},
	pages={2232--2246},
	publisher={IEEE}
}

@article{zhang2025channel,
	title={Channel measurements and modeling for dynamic vehicular {ISAC} scenarios at 28 {GHz}},
	author={Zhang, Zhengyu and He, Ruisi and Ai, Bo and Yang, Mi and Zhang, Xuejian and Qi, Ziyi and Yuan, Yuan},
	journal={IEEE Trans. Commun.},
	volume={73},
	number={8},
	pages={6884--6897},
	year={2025},
	publisher={IEEE}
}

@article{srinivasan2021airplane,
	title={Airplane-aided integrated next-generation networking},
	author={Srinivasan, Muralikrishnan and Gopi, Sarath and Kalyani, Sheetal and Huang, Xiaojing and Hanzo, Lajos},
	journal={IEEE Trans. Veh. Technol.},
	volume={70},
	number={9},
	pages={9345--9354},
	year={2021},
	publisher={IEEE}
}

@article{lu2020positioning,
	title={Positioning-aided {3D} beamforming for enhanced communications in mmWave mobile networks},
	author={Lu, Yi and Koivisto, Mike and Talvitie, Jukka and Valkama, Mikko and Lohan, Elena Simona},
	journal={IEEE Access},
	volume={8},
	pages={55513--55525},
	year={2020},
	publisher={IEEE}
}

@book{milligan2005modern,
  title={Modern antenna design},
  author={Milligan, Thomas A},
  year={2005},
  publisher={John Wiley \& Sons}
}

@article{friis1946note,
	title={A note on a simple transmission formula},
	author={Friis, Harald T},
	journal={Proceedings of the IRE},
	volume={34},
	number={5},
	pages={254--256},
	year={1946},
	publisher={IEEE}
}

@book{skolnik1980introduction,
	title={Introduction to radar systems},
	author={Skolnik, Merrill Ivan and others},
	volume={3},
	year={1980},
	publisher={McGraw-hill New York}
}

@article{hyun2021adaptive,
	title={Adaptive beam design for {V2I} communications using vehicle tracking with extended {Kalman filter}},
	author={Hyun, Seong-Hwan and Song, Jiho and Kim, Keunwoo and Lee, Jong-Ho and Kim, Seong-Cheol},
	journal={IEEE Trans. Veh. Technol.},
	volume={71},
	number={1},
	pages={489--502},
	year={2021},
	publisher={IEEE}
}

@article{jajszczyk2001digital,
  title={Digital Communication Over Fading Channels: A Unified Approach to Performance Analysis},
  author={Jajszczyk, Andrzej and Maseng, Torleiv},
  journal={IEEE Commun. Mag.},
  volume={39},
  number={7},
  pages={36--36},
  year={2001}
}

@article{zhu2018new,
  title={A new asymptotic analysis technique for diversity receptions over correlated lognormal fading channels},
  author={Zhu, Bingcheng and Cheng, Julian and Yan, Jun and Wang, Jin-yuan and Wu, Lenan and Wang, Yongjin},
  journal={IEEE Trans. Commun.},
  volume={66},
  number={2},
  pages={845--861},
  year={2018},
  publisher={IEEE}
}

@article{lopez2015asymptotically,
  title={Asymptotically exact approximations for the symmetric difference of generalized {Marcum} {$Q$}-functions},
  author={Lopez-Martinez, F Javier and Romero-Jerez, Juan M},
  journal={IEEE Trans. Veh. Technol.},
  volume={64},
  number={5},
  pages={2154--2159},
  year={2015},
  publisher={IEEE}
}

@book{abramowitz1965handbook,
	title={Handbook of mathematical functions: with formulas, graphs, and mathematical tables},
	author={Abramowitz, Milton and Stegun, Irene A},
	volume={55},
	year={1965},
	pages={376},
	publisher={Courier Corporation}
}

@article{mahony2012multirotor,
	title={Multirotor aerial vehicles: Modeling, estimation, and control of quadrotor},
	author={Mahony, Robert and Kumar, Vijay and Corke, Peter},
	journal={IEEE Rob. Autom. Mag.},
	volume={19},
	number={3},
	pages={20--32},
	year={2012},
	publisher={IEEE}
}

@article{wang2021joint,
	title={Joint beam training and positioning for intelligent reflecting surfaces assisted millimeter wave communications},
	author={Wang, Wei and Zhang, Wei},
	journal={IEEE Trans. Wireless Commun.},
	volume={20},
	number={10},
	pages={6282--6297},
	year={2021},
	publisher={IEEE}
}

@article{zhang2024design,
	title={Design and outage analysis for {VLP}-assisted indoor laser communication systems},
	author={Zhang, Kehan and Zhang, Zaichen and Zhu, Bingcheng and Chen, Shengjian and Dang, Jian and Wu, Liang and Wang, Lei},
	journal={IEEE Trans. Commun.},
	volume={72},
	number={6},
	pages={3495--3510},
	year={2024},
	publisher={IEEE}
}

@article{zhu2022outage,
	title={Outage analysis and beamwidth optimization for positioning-assisted beamforming},
	author={Zhu, Bingcheng and Zhang, Zaichen and Cheng, Julian},
	journal={IEEE Commun. Lett.},
	volume={26},
	number={7},
	pages={1543--1547},
	year={2022},
	publisher={IEEE}
}

@article{liu2026joint,
	title={Joint Gaussian Beam Pattern and Its Optimization for Positioning-Assisted Systems},
	author={Liu, Yuanbo and Zhu, Bingcheng and Huang, Shuojin and Zhang, Han and Zhang, Zaichen},
	journal={arXiv: 2603.03940},
	publisher={IEEE}
}

\end{document}